\documentclass[aps,pra,amsmath,amssymb,superscriptaddress,10pt,reprint,longbibliography]{revtex4-2}
\usepackage{graphicx,mathrsfs,enumerate}
\usepackage[utf8]{inputenc}
\usepackage{mathtools}
\usepackage{amsfonts,amstext}
\usepackage{amsmath,amsthm,amssymb,mathbbol,tensor,diagbox,multirow}
\usepackage{graphicx,xcolor,mathrsfs,enumerate}
\usepackage[hidelinks]{hyperref}\hypersetup{pdfpagemode=UseNone,pdfstartview=}
\usepackage[left=2cm,right=2cm]{geometry}
\usepackage{xcolor}
\usepackage{url}
\usepackage{times}
\usepackage{bbm}
\usepackage{soul}
\usepackage[english]{babel}

\usepackage{chngcntr}

\makeatletter
\def\maketitle{
\@author@finish
\title@column\titleblock@produce
\suppressfloats[t]}
\makeatother

\newcommand{\bra}[1]{\langle #1|}
\newcommand{\ket}[1]{|#1\rangle}
\newcommand{\braket}[2]{\langle #1|#2\rangle}

\DeclareMathOperator{\Tr}{Tr}

\newtheorem{dfn}{Definition}
\newtheorem{thm}{Theorem}

\newcommand{\mspc}{\hphantom{-}}

\begin{document}
\title{Exposure of subtle multipartite quantum nonlocality}
\author{M. M. Taddei}
\email{marciotaddei@gmail.com}
\affiliation{Federal University of Rio de Janeiro, Caixa Postal 68528, Rio de Janeiro, RJ 21941-972, Brazil}
\affiliation{ICFO - Institut de Ciencies Fot\`oniques, The Barcelona Institute of Science and Technology, 08860, Castelldefels, Barcelona, Spain}

\author{T. L. Silva}
\affiliation{Federal University of Rio de Janeiro, Caixa Postal 68528, Rio de Janeiro, RJ 21941-972, Brazil}

\author{R. V. Nery}
\affiliation{Federal University of Rio de Janeiro, Caixa Postal 68528, Rio de Janeiro, RJ 21941-972, Brazil}
\affiliation{International Institute of Physics, Federal University of Rio Grande do Norte, 59070-405, Natal, Brazil}

\author{G. H. Aguilar}
\affiliation{Federal University of Rio de Janeiro, Caixa Postal 68528, Rio de Janeiro, RJ 21941-972, Brazil}

\author{S. P. Walborn}
\affiliation{Federal University of Rio de Janeiro, Caixa Postal 68528, Rio de Janeiro, RJ 21941-972, Brazil}
\affiliation{Departamento de F\'{\i}sica, Universidad de Concepci\'on, 160-C Concepci\'on, Chile}
\affiliation{ANID – Millennium Science Initiative Program – Millennium Institute for Research in Optics, Universidad de Concepci\'on, 160-C Concepci\'on, Chile}

\author{L. Aolita}
\affiliation{Federal University of Rio de Janeiro, Caixa Postal 68528, Rio de Janeiro, RJ 21941-972, Brazil}
\affiliation{Quantum Research Centre, Technology Innovation Institute, Abu Dhabi, UAE} 

\date{\today}
\begin{abstract}
The celebrated Einstein-Podolsky-Rosen quantum steering has a complex structure in the multipartite scenario. We show that a naively defined criterion for multipartite steering allows, like in Bell nonlocality, for a contradictory effect whereby local operations could create steering seemingly from scratch. Nevertheless, neither in steering nor in Bell nonlocality has this effect been experimentally confirmed.
Operational consistency is reestablished by presenting a suitable redefinition: there is a subtle form of steering already present at the start, and it is only exposed --- as opposed to created --- by the local operations. We devise protocols that, remarkably, are able to reveal, in seemingly unsteerable systems, not only steering, but also Bell nonlocality. Moreover, we find concrete cases where entanglement certification does not coincide with steering. A causal analysis reveals the crux of the issue to lie in hidden signaling.
Finally, we implement one of the protocols with three photonic qubits deterministically, providing the experimental demonstration of both exposure and super-exposure of quantum nonlocality.
\end{abstract}

\maketitle


\section*{Introduction}\label{sec:intro}

Three forms of quantum correlation occur in nature---entanglement, Bell nonlocality and steering.  
The distinction between them can be viewed, from an operational perspective, as given by the level of trust and control that one has on the systems involved. 
Entanglement, for instance, is naturally formulated in the so-called device-dependent (DD) scenario \cite{Horodecki2009}. There, one assumes that the system can be completely characterized by the measurement apparatus, 
at least in principle. Bell nonlocality, in contrast, takes place in the device-independent (DI) description \cite{Brunner2014}. There, measurement devices are treated as untrusted black boxes whose actual measurement process is uncharacterized or ignored, relying only on classical measurement settings (inputs) and results (outputs). 
Quantum steering, on the other hand, is a hybrid type of correlation -- intermediate between entanglement and Bell nonlocality -- that arises in semi-DI settings \cite{Reid2009,Cavalcanti2017,Uola2019}. The latter involve both DD and DI parties, and an example is shown in Fig.\ \ref{fig:wiring}\textbf{a} for the tripartite case of two untrusted devices and one trusted one. For all three types of correlation, the multipartite scenario is considerably richer than the bipartite one. 

Whereas entanglement is a resource for DD applications in quantum information, Bell nonlocality is the key resource for DI applications such as DI quantum key distribution \cite{Barrett2005,Acin2006,Acin2006a,Acin2007}, DI-certified randomness \cite{Colbeck2009,Colbeck2010,Pironio2010,Acin2016}, DI-verifiable blind quantum computation \cite{Gheorghiu2015,Hajdusek2015} and DI conference-key agreement \cite{Ribeiro2018,Holz2020,Murta2020}, which are typically much more experimentally demanding than the corresponding DD protocols. Steering is known to be the crucial resource for key technological applications in the semi-DI scenario, which are generally less technically difficult than their DI counterparts, while requiring less assumptions than the corresponding DD protocols. These include semi-DI entanglement certification \cite{Wiseman2007,Jones2007,Cavalcanti2017,Uola2019}, quantum key distribution \cite{Branciard2012,He2013}, certified-randomness generation \cite{Skrzypczyk2018}, quantum secret sharing \cite{Kogias2017,Xiang2017}, as well as other useful protocols in multipartite quantum networks \cite{Huang2019}.

\begin{figure*}[t]
	\centering
	\includegraphics[width=1.8\columnwidth]{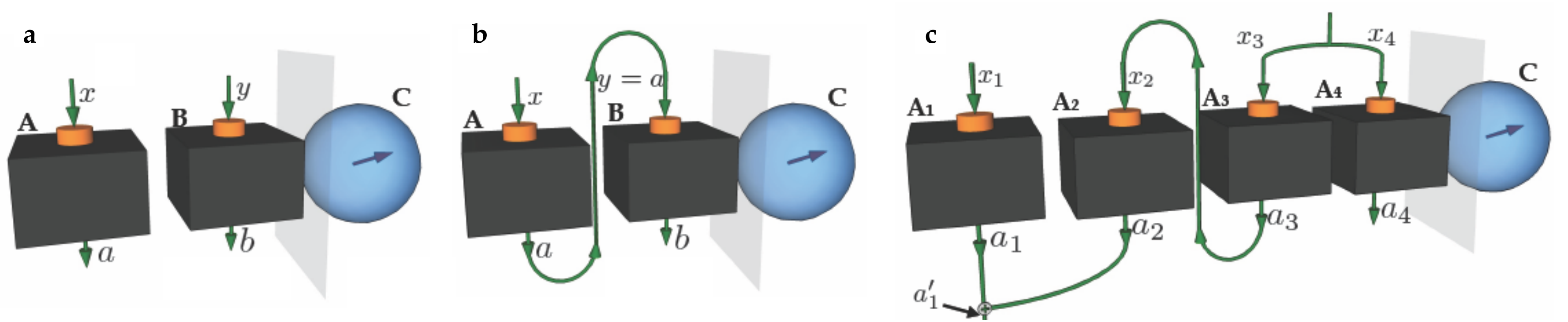}
	\caption{Several hybrid (trusted-untrusted) multipartite scenarios. In the device-dependent (DD) case, measurement devices are well characterized (trusted), so that a specific quantum state (represented by Bloch spheres) can be attributed to the system. In the device-independent (DI) case, in contrast, the devices are uncharacterized (untrusted), so that systems are represented by black boxes. Semi-DI scenarios contain both trusted and untrusted components. There, the joint system is mathematically described by a hybrid object -- intermediate between a  state and a Bell behavior -- called \emph{assemblage}, and the type of nonlocality they can feature is called \emph{steering}. In all three panels the shaded plane illustrates the bipartition of the trusted subsystem versus the untrusted ones. \textbf{a} An assemblage in the 2DI+1DD scenario: Alice  and Bob rely on a black-box description, whereas Charlie's system is trusted. All three subsystems are space-like separated. \textbf{b} Alice and Bob are no longer space-like separated: she communicates her output to him and he uses this to choose his input. This is an example of a bilocal \emph{wiring} (local with respect to the bipartition $AB\vert C$). Such operations cannot create any correlations across the bipartition, but they can \emph{expose} a subtle form of multipartite quantum nonlocality that otherwise does not violate any Bell or steering inequality across the bipartition (see text). \textbf{c} A 4DI+1DD assemblage is mapped onto a 2DI+1DD one by a bilocal wiring ($x_2=a_3$, $x_3=x_4$, and $a'_1=a_1+a_2\mod2$). Such wirings can implement non-trivial resource-theoretic transformations, but not enough to enable a multi-black-box universal \emph{steering bit}, i.e. an $N$-partite assemblage from which all bipartite ones, e.g., can be reached (see  Supplementary Notes \ref{app:proof_nobits}).\\}
	\label{fig:wiring}
\end{figure*}

Interestingly, an operational inconsistency has arisen in the fully DI multipartite scenario \cite{Gallego2012,Bancal2013}. It is rooted in the existence of an operation local to the $AB$ partition that can create a Bell nonlocality across $AB|C$. The issue, however, is best understood with the framework of resource theories.

Resource theories constitute formal treatments of a physical property as a resource, providing a complete toolbox for its quantification, classification, and operational manipulation (see, e.g., \cite{Brandao2015,Brandao2015a,Coecke2016}). Applied and fundamental interest has motivated their formulation for entanglement \cite{Horodecki2009} and Bell nonlocality \cite{Gallego2012,deVicente2014,Gallego2017,Wolfe2019}, as well as for other relevant quantum properties \cite{Winter2016,Chitambar2016,Grudka2013,Amaral2018,Taddei2019,Wolfe2019}. Most important for our discussion is the resource theory of steering \cite{Gallego2015,Kaur2017}. The cornerstone of any resource theory is the set of its \emph{free operations}. These are unable to create the resource: they transform every resourceless state into a resourceless state. As a concrete example, free operations for quantum steering include, on the untrusted side, pre and post-processings of classical variables of the black boxes and, on the trusted side, local quantum operations and classical communication to the untrusted parties. It can be shown \cite{Gallego2015} that these operations cannot create quantum steering out of unsteerable systems.
 
A fully DI description is cast in terms of a \emph{Bell behavior}, given by a conditional probability distribution of the outputs given the inputs. Bell locality implies that there exists a local-hidden-variable (LHV) model, in which correlations are explained by a (hypothetical) classical common cause (the hidden variable) within the common past light-cone of the measurement events \cite{Bell1964}. Any Bell-inequality violation implies incompatibility with LHV models, i.e. Bell nonlocality. Bell-local behaviors are, naturally, the resourceless states of the resource theory of Bell nonlocality. We shall use the term \emph{bilocal} to refer to being local with respect to the $AB|C$ bipartition.  
It stands to reason that operations within a given partition are free. However, a ``wiring'' between $A$ and $B$ (e.g.\ linking the output of one black box to the input of another as in Fig.~1 \textbf{b}) is confined to $AB$ but can map tripartite Bell behaviors that are local in the $AB|C$ partition (i.e. bilocal) into bipartite Bell behaviors that violate a Bell inequality across $AB|C$. 
The problem, however, lied in the  definition of Bell nonlocality in multipartite scenarios used previously \cite{Svetlichny1987}.

According to the traditional definition \cite{Svetlichny1987}, Bell nonlocality across a system bipartition is incompatible with any LHV model with respect to it. This includes so-called ``fine-tuned'' models \cite{Wood2015} with hidden signaling. These are LHV models where, for each value of the hidden variable, the subsystems on each side of the bipartition communicate, but for which the statistical mixture over all values of the hidden variable renders the observable correlations non-signaling. 
The problem is that the bilocal wiring can conflict with the hidden signaling in such models, giving rise to a causal loop. 
For instance, assume that for a particular tripartite system, there is only one LHV decomposition, which uses hidden signaling from Bob to Alice. To physically implement the wiring in Fig.\ \ref{fig:wiring}\textbf{b}, which is an example of a free operation allowed within the AB partition, Bob must be in the causal future of Alice. This, in turn, is inconsistent with the direction of the hidden signaling. 
This explains why apparently bilocal behaviors can lead to Bell violations after a bilocal wiring.  A redefinition of multipartite Bell nonlocality was then proposed \cite{Gallego2012,Bancal2013}. This considers the correlations from conflicting bilocal models as already nonlocal across the bipartition, so that the wiring simply exposes an already-existing subtle form of Bell nonlocality. We refer to the latter form and effect as subtle Bell nonlocality and Bell-nonlocality exposure, respectively.

The redefinition fixed the inconsistency, but also opened several intriguing questions. First, no experimental observation of Bell-nonlocality exposure has been reported. Second, even though steering theory is relatively mature \cite{Cavalcanti2011,He2013,Armstrong2015,Taddei2016,Li2015}, little is known about \emph{steering exposure}. 
Steering features in the semi-DI description, where systems are described in terms of \emph{assemblages}, given by quantum states describing the DD subsystems, weighted by the conditional probabilities describing the DI parties. Operational consistency relative to steering exposure was considered, in particular, 
in a  definition of multipartite steering \cite{He2013}, but based on models where each party is probabilistically either trusted or untrusted. On the other hand,  a definition based on multipartite entanglement certification in semi-DI setups with fixed trusted-versus-untrusted divisions was proposed in Ref. \cite{Cavalcanti2015a}. There, bilocal hidden-variable  models with an explicit quantum realization are considered, which automatically rules out potentially-conflicting fined-tuned models. Nevertheless, this has the side-effect of over-restricting the set of unsteerable assemblages, thus potentially over-estimating steering. Third, exposure as a resource-theoretic transformation is yet unexplored territory. For instance, is it possible to obtain every bipartite assemblage via exposure from some multipartite one? What about Bell behaviors? Moreover, is there a single $N$-partite assemblage from which all bipartite ones are obtained via exposure?

These are the questions we answer. To begin with, we show that, remarkably, exposure of quantum nonlocality is a universal effect, in the sense that every bipartite Bell behavior (assemblage) can be the result of Bell-nonlocality (steering) exposure starting from some tripartite one. This highlights the power of exposure as a resource-theoretic transformation. However, we also delimit such power: we prove a no-go theorem for multi-black-box universal steering bits: there exists no single $N$-partite assemblage (with $N-1$ untrusted and 1 trusted devices) from which all bipartite ones can be obtained through free operations of steering. Interestingly, in the universal steering exposure protocol, the starting behavior is not guaranteed to admit a physical realization, i.e. it may be supra-quantum \cite{Sainz2015,Sainz2018a,Sainz2019}. Therefore, we also derive an alternative protocol that -- albeit not universal -- is manifestly within quantum theory. 
Moreover, we show that the output assemblage of such protocol is not only steerable but also Bell nonlocal (in the sense of producing a nonlocal behavior upon measurements by Charlie). This is notable as Bell nonlocality is a stronger form of quantum correlation than steering. We refer to this effect as \emph{super-exposure of Bell nonlocality}. 
In turn, we provide a redefinition of (both multipartite and genuinely multipartite) steering to re-establish operational consistency. 
Finally, we experimentally demonstrate exposure as well as super-exposure. This is done using polarization and path degrees of freedom of two entangled photons generated by spontaneous parametric down conversion, in a deterministic protocol. 

\section*{Preliminaries}
\subsection*{Steering and the semi-DI setting} \label{sec:prelim}
Most of our discussion will be based on the semi-DI setting of Fig.\ \ref{fig:wiring}\textbf{a}. We will not resort to quantum models of the black boxes; our definitions are based on the semi-DI setting alone, as befits its treatment as a resource for quantum tasks. Such systems are fully described by a Bell behavior $\boldsymbol{P}^{(AB)}:=\{P_{a,b|x,y}\}_{a,b,x,y}$, with $P_{a,b|x,y}$ the conditional probability of outputs $a,b$ given inputs $x,y$, for Alice and Bob, and an ensemble of conditional quantum states $\varrho_{a,b|x,y}$ for Charlie. These can be encapsulated in a hybrid object known as the assemblage $\boldsymbol\sigma:=\{\sigma_{a,b|x,y}\}_{a,b,x,y}$, of sub-normalized conditional states $\sigma_{a,b|x,y}:=P_{a,b|x,y}\,\varrho_{a,b|x,y}$. 

Unlike in Bell nonlocality or entanglement, semi-DI systems have a natural bipartition: the one separating the trusted devices from the untrusted ones, $AB|C$. This is the bipartition with respect to which we define steering throughout, unless otherwise explicitly stated.
We assume that $\boldsymbol\sigma$ satisfies the no-signaling (NS) principle, by virtue of which measurement-outcome correlations alone do not allow for communication. This implies that the statistics observed by Charlie should be independent of the input(s) of the remaining user(s). Mathematically, this condition reads 
\begin{equation}
\sum_{a,b}\sigma_{a,b|x,y} = \varrho^{(C)}, \ \quad\forall\ x, y,
\label{eq:NSAB-C}
\end{equation}
where $\varrho^{(C)}$ is the reduced state on $C$.  Furthermore, we also assume that Alice and Bob are NS, i.e.\ choosing their inputs does not provide them any communication,
\begin{align}
\sum_a\sigma_{a,b|x,y}& \ \mbox{ independent of } x,\ \quad\forall\ b, y, \label{eq:NSAC-B}\\
\sum_b\sigma_{a,b|x,y}& \ \mbox{ independent of } y,\ \quad\forall\ a, x. \label{eq:NSA-BC}
\end{align}
 
The definition of steering in the $AB|C$ partition hinges on the impossibility of decomposing an assemblage $\boldsymbol\sigma$ as
\begin{equation}
\sigma_{a,b|x,y} = \sum_\lambda \ P_{\lambda}\ \ P_{a,b|x,y,\lambda} \,\varrho_{\lambda} \ .
\label{eq:LHS}
\end{equation}
Here, $P_{\lambda}$ is the probability of the hidden variable $\Lambda$ taking the value $\lambda$, each $\boldsymbol{P}_\lambda^{(AB)}:=\{P_{a,b|x,y,\lambda}\}_{a,b,x,y}$ is a $\lambda$-dependent behavior, and $\varrho_{\lambda}$ is the $\lambda$-th hidden state for $C$ (locally correlated with $AB$ only via $\Lambda$). However, different approaches have diverging positions on the set to which the distribution $\boldsymbol{P}_\lambda^{(AB)}$ may belong. Possibilities range \cite{Uola2019} from the full set of valid bipartite distributions to the most restricted set of factorizable ones (i.e. $P_{a,b|x,y,\lambda}=P_{a|x,\lambda}P_{b|y,\lambda}$ $\forall a,b,x,y$). In \cite{Cavalcanti2015a}, steering is treated as equivalent to entanglement certification, hence each distribution $\boldsymbol{P}_\lambda^{(AB)}$ is required to be quantum-mechanically realizable.
Our operational approach is defined in terms of assemblages only and aims to use them as resources, not for inferences on the quantum models that can produce them. It is thus best to ignore restrictions and consider, as a starting point, a general probability distribution. As such, $\boldsymbol\sigma$ is unsteerable if it admits a local hidden-state (LHS) model, defined by Eq.\eqref{eq:LHS} with general $\boldsymbol{P}_\lambda^{(AB)}$; otherwise  $\boldsymbol\sigma$ is steerable.

Importantly, a non-signaling $\boldsymbol\sigma$ does not imply non-signaling $\boldsymbol{P}_\lambda^{(AB)}$ for each $\lambda$. (Imposition of the latter would be an additional requirement, one that is used in \cite{Cavalcanti2017} for yet another definition of steering in the literature.) In fact, LHS models can exploit hidden signaling between Alice and Bob as long as communication at the observable level (i.e. upon averaging $\Lambda$ out) is impossible. This effect is known as fine-tuning \cite{Wood2015} and will turn out to be critical.

\section*{Results}
\subsection*{Steering exposure and Bell-nonlocality super-exposure}
We begin by an exposure protocol for steering and Bell nonlocality that is universal in the sense of being capable of producing any bipartite assemblage (behavior) whatsoever from an appropriate tripartite assemblage  (behavior) originally admitting an LHS (LHV) model. As in Ref. \cite{Gallego2012}, we exploit bilocal wirings as that of Fig.\,\ref{fig:wiring}\textbf{b}, which makes Bob's input $y$ equal to Alice's output $a$. This requires that Bob's measurement is in the causal future of Alice's. Indeed, after the wiring, systems $A$ and $B$ now behave as a single black box with input $x$ and output $b$. In other words, exposure is a form of conversion from tripartite correlations into bipartite ones. Here, we restrict to the case of binary inputs and outputs ($x$, $y$, $a$, $b$ $\in\{0,1\}$) for simplicity, where we prove the following surprising result.

\begin{thm}[Universal exposure of quantum nonlocality]\label{th:universal}
Any bipartite assemblage $\boldsymbol\sigma^{(\text{target})}$ or Bell behavior $\boldsymbol{P}^{(\text{target})}$ can be obtained via the wiring $y=a$ on the tripartite assemblage $\boldsymbol\sigma^{(\text{initial})}$ or behavior $\boldsymbol{P}^{(\text{initial})}$, respectively, of elements
\begin{equation}
\sigma_{a,b|x,y}^{(\text{initial})} := \frac12 \sigma_{b|x\oplus a\oplus y}^{(\text{target})} \ 
\label{eq:genactiv}
\end{equation}
\textit{or}
\begin{equation}
P^{(\text{initial})}(a,b,c|x,y,z) = \frac12 P^{(\text{target})}(b,c|x\oplus a\oplus y,z) \ ,
\label{eq:genblackboxactiv}
\end{equation}
where $\oplus$ stands for addition modulo 2. Moreover, $\boldsymbol\sigma^{(\text{initial})}$ and $\boldsymbol{P}^{(\text{initial})}$ admit respectively an LHS and an LHV models across the $AB|C$ bipartition, for all $\boldsymbol\sigma^{(\text{target})}$ and $\boldsymbol{P}^{(\text{target})}$.
\end{thm}

That the initial correlations are mapped to the desired target is self-evident from Eqs.\ (\ref{eq:genactiv},\ref{eq:genblackboxactiv}). What is certainly not evident is that the initial correlations are bilocal. This is proven in  Supplementary Notes \ref{app:proof_exposure_universal} by construction of explicit LHS and LHV models. When the target assemblage (behavior) is steerable (Bell nonlocal), exposure of steering (Bell nonlocality) is achieved. Furthermore, apart from steerable, assemblages can also be Bell nonlocal in the sense of giving rise to nonlocal behaviors under local measurements \cite{Taddei2016}. Hence, when $\boldsymbol\sigma^{(\text{target})}$ is Bell nonlocal, a seemingly unsteerable system --- i.e.\ one that admits an LHS decomposition --- is mapped onto a Bell nonlocal one, which is outstanding in view of the fact that unsteerable assemblages form a strict subset of Bell-local ones. 

The protocol highlights the capabilities of bilocal wirings as resource-theoretic transformations. Remarkably, such wirings compose a strict subset of well-known classes of free operations of quantum nonlocality (across $AB|C$): local operations with classical communication (LOCCs) \cite{Horodecki2009} for entanglement, one-way (1W) LOCCs from the trusted to the untrusted parts \cite{Gallego2015} for steering, and local operations with shared randomness \cite{Gallego2012,deVicente2014,Gallego2017} for Bell nonlocality. However, there are also limitations to the capabilities of these wirings. In particular, in Supplementary Notes \ref{app:proof_nobits} we prove a no-go theorem for universal steering bits in the $N$DI+1DD scenario [exemplified in Fig.\ \ref{fig:wiring}\textbf{c} for $N=4$]. That is, we show there that there is no $N$-partite assemblage, for all $N$, from which all bipartite ones can be obtained via arbitrary 1W-LOCCs.

Although the protocol above is universal, it is unclear whether it can actually be physically implemented in general. This is due to the fact that the tripartite initial correlations may be supra-quantum, i.e.\ well-defined non-signaling correlations that can however not be obtained from local measurements on any quantum state \cite{Sainz2015,Sainz2018a,Sainz2019,Popescu1994}. Physical protocols for Bell-nonlocality exposure were devised in Refs.\ \cite{Gallego2012,Bancal2013}, but no such protocols have been reported for steering. Hence, we
next  derive an alternative example for both steering exposure and Bell-nonlocality super-exposure that is manifestly within quantum theory. This also exploits the bilocal wirings of Fig.~\ref{fig:wiring}\textbf{b}, but starting from a different initial assemblage.
We describe the latter directly in terms of its quantum realization.
Consider a tripartite Greenberg-Horne-Zeilinger (GHZ) state $(|000\rangle+|111\rangle)/\sqrt2$, with $\ket0$ and $\ket1$ the eigenvectors of the third Pauli matrix $Z$. Bob makes von Neumann measurements on his share of the state for both his inputs, for $y=0$ in the $Z+X$ basis and for $y=1$ in the $Z-X$ basis, with $X$ the first Pauli matrix.  Alice, however, makes either a trivial measurement, given by the positive operator-valued measure $\{\mathbb 1/2,\mathbb 1/2\}$, for $x=0$, or a von Neumann $X$-basis measurement, for $x=1$. For the resulting initial
 assemblage, $\boldsymbol\sigma^{(\text{GHZ})}$, the following holds (see Supplementary Notes \ref{app:proof_exposure_universal_quantum} for more details).

\begin{thm}[Physically-realizable exposure and super-exposure]\label{th:quantum_protocol}
The quantum assemblage $\boldsymbol\sigma^{(\text{GHZ})}$, of elements
\begin{equation}
\sigma^{(\text{GHZ})}_{a,b|x,y}
 = \frac18 \left\{\mathbb1 + \frac{(-1)^b}{\sqrt2}\left[Z + x(-1)^{a+y}X\right]\right\} \ 
\label{eq:quantumorig}
\end{equation}
admits an LHS model and, under the wiring $y=a$, is mapped to the assemblage of elements
\begin{equation}
\sigma_{b|x}
 = \frac14 \left[\mathbb1 + \frac{(-1)^b}{\sqrt2}\left(Z + x X\right)\right] \ ,
\label{eq:quantumsteerable}
\end{equation}
which is both steerable and Bell-nonlocal.
\end{thm}

These results require a redefinition of steering in the multipartite scenario, since an assemblage can admit an LHS decomposition and still be steerable. We describe this redefinition, analogous to the one in \cite{Gallego2012}, before moving on to the experimental realization.

\subsection*{Consistently defining steering}\label{sec:def_steering}

The existence of subtle steering implies a stark inconsistency between the naive definition of steering from LHS decomposability, Eq.~\eqref{eq:LHS}, and the notion of locality. Since the free operations that cause exposure are classical and strictly local (fully contained in the $AB$ partition), it is reasonable that they are unable to create not only steering but also any form of correlations (even classical ones) across  $AB|C$. The alternative left is to redefine bipartite steering in multipartite scenarios such that, e.g., the assemblages in Eqs.\ \eqref{eq:genactiv} and \eqref{eq:quantumorig} are already steerable. Formally, we need to exclude a subclass of LHS decompositions from the set of unsteerable assemblages.

To identify that subclass, let us apply the wiring $y=a$ to a general $\boldsymbol\sigma$ fulfilling Eq.\ \eqref{eq:LHS}. This gives $\boldsymbol\sigma^{(\text{wired})}$, of elements
\begin{align}
\sigma_{b|x}^{(\text{wired})}&:=\sum_a\sigma_{a,b|x,a}= \sum_\lambda \ \ P_{\lambda} \Big(\sum_a P_{a,b|x,a,\lambda}\Big) \varrho_{\lambda}.
\label{eq:wiredLHS} 
\end{align}
This is a valid LHS decomposition as long as the term within parentheses yields a valid (normalized) conditional probability distribution (of $B$ given $X$ and $\Lambda$). This is the case if every $\boldsymbol{P}_\lambda^{(AB)}$ in Eq.\ \eqref{eq:LHS} is non-signaling. In that case, by summing over $b$ and applying the NS condition, one gets
\begin{equation}
 \sum_{a,b}P_{a,b|x,a,\lambda} = \sum_{a}P_{a|x,a,\lambda} \stackrel{\rm NS}{=} \sum_a P_{a|x,\lambda} = 1 \ ,
\label{eq:norm}
\end{equation}
which renders $\boldsymbol\sigma^{(\text{wired})}$ indeed unsteerable. However, this reasoning can in general not be applied if any $\boldsymbol{P}_\lambda^{(AB)}$ is signaling from Bob to Alice, i.e. if Alice's marginal distribution for $a$ depends on $y$ (apart from  $x$ and $\lambda$).
\begin{figure}%
\includegraphics[width=.8\columnwidth]{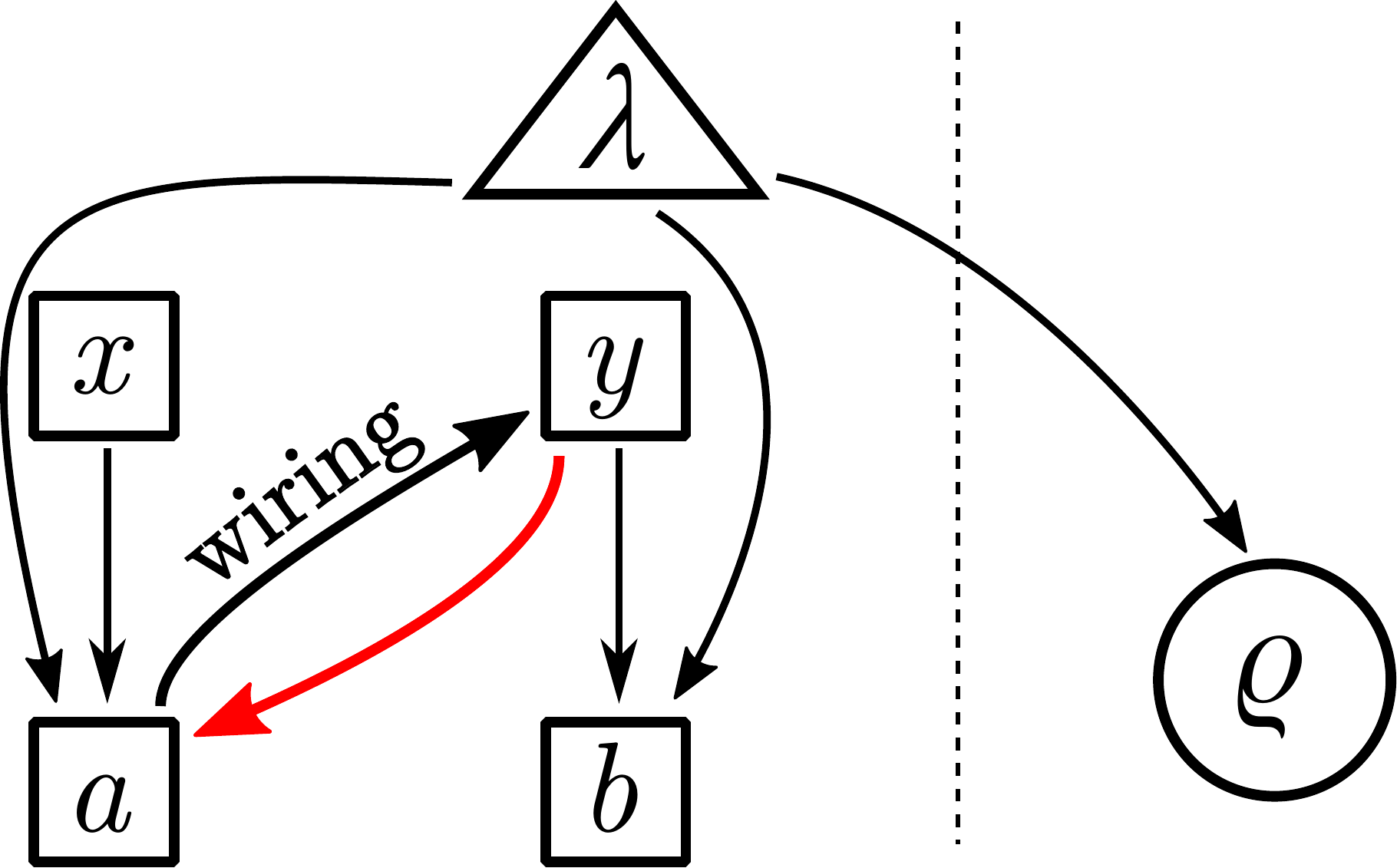}%
\caption{ Steering exposure as a causal loop. In the causal network underlying LHS models, given by Eq.\ \eqref{eq:LHS}, the hidden variable $\lambda$ directly influences Charlie's quantum state $\varrho$ as well as the Alice and Bob's outputs $a$ and $b$, which are in turn also influenced by the inputs $x$ and $y$, respectively. Even though the observed assemblage (after averaging  $\lambda$ out) is non-signaling, the model can still exploit hidden signaling (i.e. at the level of $\lambda$). For instance, for each $\lambda$, Alice's output may depend (red arrow) on Bob's input in a different fine-tuned way such that the dependence vanishes at the observable level.
The wiring of Fig.\ \ref{fig:wiring}\textbf{b} forces $y=a$, closing a causal loop that will in general conflict with the latter dependence for some $\lambda$.
As a consequence, the final assemblage resulting from the wiring may not admit a valid LHS decomposition, exposing steering. Hence, the exposure can in a sense be thought of as an operational benchmark for hidden signaling in the LHS model describing the initial assemblage.
}%
\label{fig:causal}%
\end{figure}
Therefore, we see that the inconsistency is rooted at hidden signaling.
In fact, at the level of the underlying causal model, the phenomenon of exposure can be understood as a causal loop between such signaling and the applied wiring (see Fig.\ \ref{fig:causal}).

To restore consistency, hidden signaling must be restricted. An obvious possibility would be to allow only for non-signaling $\boldsymbol{P}_\lambda^{(AB)}$'s in Eq.\ \eqref{eq:LHS}. Interestingly, however, this turns out to be over-restrictive. Following the redefinition of multipartite Bell nonlocality \cite{Gallego2012,Bancal2013}, we propose the following for bipartite steering in multipartite scenarios.

\begin{dfn}[Redefinition of steering]
\label{def:TOLHS_NSLHS}
\textit{An assemblage $\boldsymbol\sigma$ is unsteerable if it admits \emph{time-ordered LHS} (TO-LHS) decompositions both from $A$ to $B$ and from $B$ to $A$ simultaneously, i.e. if}
\begin{align}
\sigma_{a,b|x,y} 
= & \sum_\lambda P_{\lambda}\ \ P^{(A\to B)}_{a,b|x,y,\lambda}\ \varrho_{\lambda} \label{eq:signAB}\\
= & \sum_\lambda P'_{\lambda} \ \ P^{(B\to A)}_{a,b|x,y,\lambda}\ \varrho'_{\lambda} \label{eq:signBA} \ ,
\end{align}
\textit{where each $\boldsymbol{P}^{(A\to B)}_\lambda$ is non-signaling from Bob to Alice and each $\boldsymbol{P}^{(B\to A)}_\lambda$ from Alice to Bob. Otherwise  $\boldsymbol\sigma$ is steerable.}
\label{eq:TOdefinition}
\end{dfn}

\begin{figure}%
\includegraphics[width=.6\columnwidth]{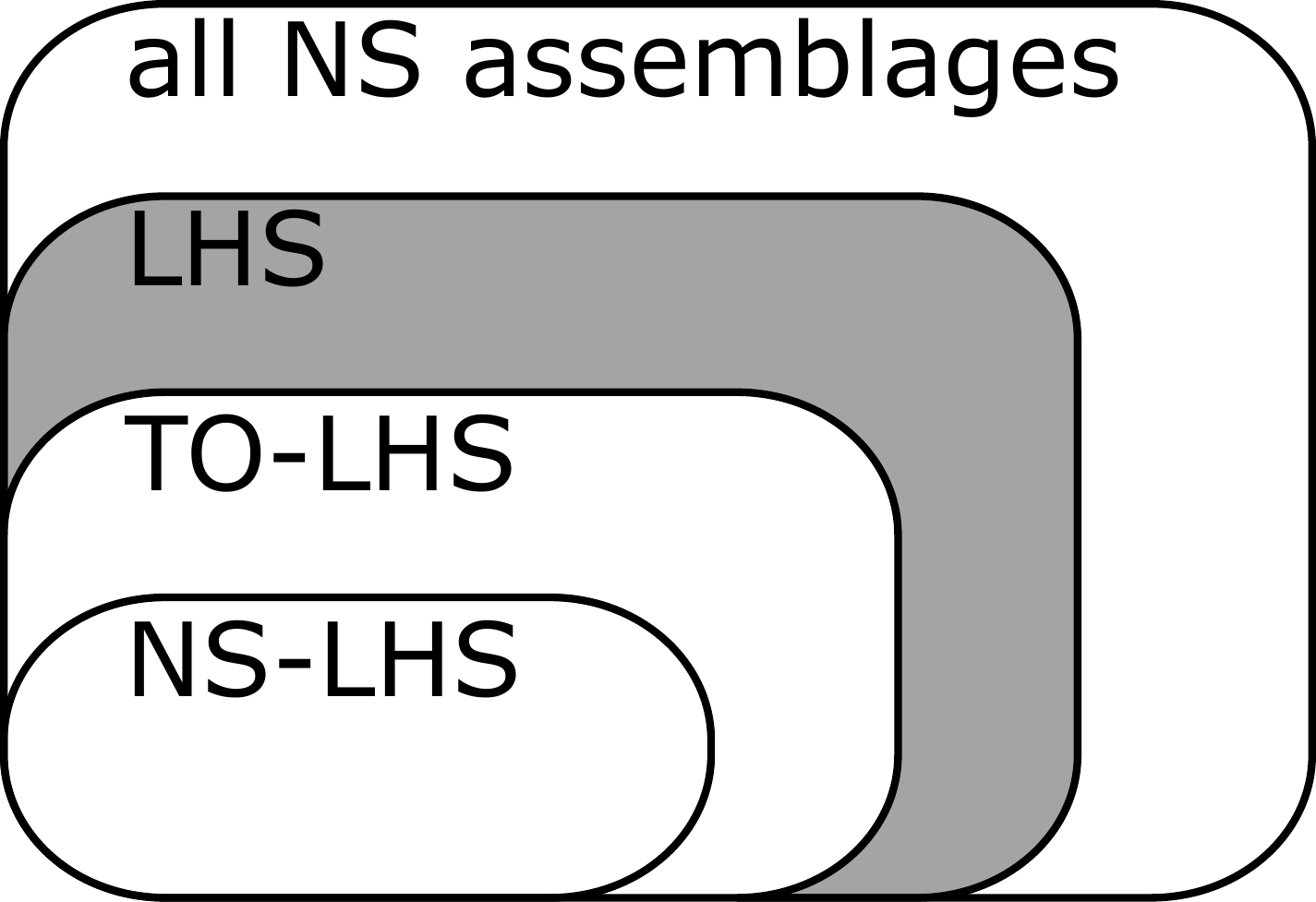}%
\caption{ 
Pictorial representation of inner structure of the set of all non-signaling assemblages in the tripartite scenario. 
Inclusion is strict for all depicted subsets: the set \textsf{LHS} of generic local-hidden-state (LHS) assemblages, the set \textsf{TO-LHS} of time-ordered LHS ones, the set \textsf{NS-LHS} of non-signaling LHS ones, and the set \textsf{Q-LHS} of quantum-LHS ones (see Supplementary Notes \ref{app:strict_subset} for details).
The shaded region represents the set of assemblages with subtle steering. Bilocal wirings can expose such steering by mapping that region to the set of (bipartite) steerable assemblages.
}%
\label{fig:setsLHS}%
\end{figure} 

The validity of both time orderings simultaneously prevents conflicting causal loops. More precisely, if a wiring from Alice to Bob is applied on $\boldsymbol\sigma$, one uses decomposition \eqref{eq:signAB} to argue with the $\boldsymbol{P}^{(A\to B)}_\lambda$'s [as in Eq.\ \eqref{eq:norm}] that the wired assemblage is unsteerable. Analogously, if a wiring from Bob to Alice is performed, one argues using the $\boldsymbol{P}^{(B\to A)}_\lambda$'s from decomposition \eqref{eq:signBA}. Hence, no exposure is possible for TO-LHS assemblages, guaranteeing consistency with bilocal wirings (as well as generic 1W-LOCCs from trusted to untrusted parts) as free operations of steering. 
We note that, even though this redefinition prevents the exposure effect from creating steering, the effect still has, as illustrated by the exposure theorems, a relevant transformation power, especially when applied to steerable assemblages.  As an example, there are assemblages that can only violate a Bell inequality across $AB|C$ after the exposure protocol.

On the other hand, when all $\lambda$-dependent behaviors in Eqs.\ (\ref{eq:signAB},\ref{eq:signBA}) are fully non-signaling, then the assemblage is called \emph{non-signaling LHS} (NS-LHS). 
There exists TO-LHS assemblages that are not NS-LHS, which proves that the latter is a strict subset of the former.
In Supplementary Notes \ref{app:strict_subset}, we provide a quantum and a supra-quantum example of TO-LHS assemblages that are not NS-LHS.

This definition based on TO-LHS models is strictly different from previous definitions of steering in the literature. In \cite{Cavalcanti2017}, $\boldsymbol P_\lambda^{(\boldsymbol{AB})}$ from Eq.\eqref{eq:LHS} is restricted to non-signaling distributions, which coincides with the NS-LHS definition. In \cite{Cavalcanti2015a} $\boldsymbol P_\lambda^{(\boldsymbol{AB})}$ is further restricted to quantum-realizable bipartite distributions, in what constitutes the quantum-LHS model, see Fig.\ \ref{fig:setsLHS}. A fully factorizable $\boldsymbol P_\lambda^{(\boldsymbol{AB})}$, as mentioned in \cite{Uola2019}, represents an even further restriction, and the corresponding model only allows for classical correlations between Alice, Bob, and Charlie.

These examples have another consequence for the definition of steering. At times has the definition of steering been stated as entanglement that can be certified with the reduced information content of a semi-DI setting \cite{Wiseman2007,Cavalcanti2015a}. In fact, even with steering defined independently of entanglement certification, never to our knowledge had there been an instance of one being present without the other. We have nevertheless found cases of entanglement certification in the semi-DI scenario without steering, dissociating these two notions: the latter is sufficient, but not necessary, for the former. This is seen from the quantum-realizable examples of a TO-LHS assemblage that is not NS-LHS in Supplementary Notes \ref{app:strict_subset}. They can be decomposed as in Eqs.\ (\ref{eq:signAB},\ref{eq:signBA}), but only with distributions $P_{a,b|x,y,\lambda}^{(A,B\to B,A)}$ that are signaling, hence, not quantum. As such, a quantum system without $AB|C$ entanglement is unable to produce such an assemblage, i.e.\ entanglement can be certified in $AB|C$. On the other hand, since it is TO-LHS, the assemblage has no steering in the same bipartition (details in Supplementary Notes \ref{app:strict_subset}).

Furthermore, the redefinition above automatically implies also a redefinition of genuinely multipartite steering (GMS). We present this explicitly in Supplementary Notes \ref{sec:def_gen_multipartite}. There, we follow the approach of Ref.~\cite{Cavalcanti2015a} in that a fixed trusted-versus-untrusted partition is kept. However, instead of defining GMS as incompatibility with quantum-LHS assemblages (i.e. with $\lambda$-dependent behaviors with explicit quantum realizations) as in \cite{Cavalcanti2015a}, we use the more general TO-LHS ones. This reduces the set of genuinely multipartite steerable assemblages safely, i.e.\ without introducing room for exposure. The dissociation of steering and entanglement certification also happens in this genuine multipartite case.

\subsection*{Experimental implementation}\label{sec:experimental}
The exposure procedure was experimentally implemented using entangled photons produced via spontaneous parametric down conversion. The experimental setup is shown in Fig.\,\ref{fig:setup}. A photon pair is generated in the Bell state $|\Phi^+\rangle = \left(|00\rangle+|11\rangle\right)/\sqrt{2}$,  where $|0\rangle$ ($|1\rangle$) stands for horizontal (vertical) polarization of the photons  \cite{Kwiat1999}. The photons in the signal mode ($s$) pass through a calcite beam displacer (BD), which creates two momentum modes (paths) depending on the polarization. This results in a tripartite  GHZ state, {where the extra qubit is the path degree of freedom of the photons in $s$.}  Alice's and Bob's qubits are the polarization and path of the photons in mode $s$, respectively, while Charlie's qubit is the polarization of the photons in mode $i$. Projective measurements onto all the degrees of freedom required for state tomography  are performed as described below.

\begin{figure}[tb]
	\centering
	\includegraphics[width=0.99\columnwidth]{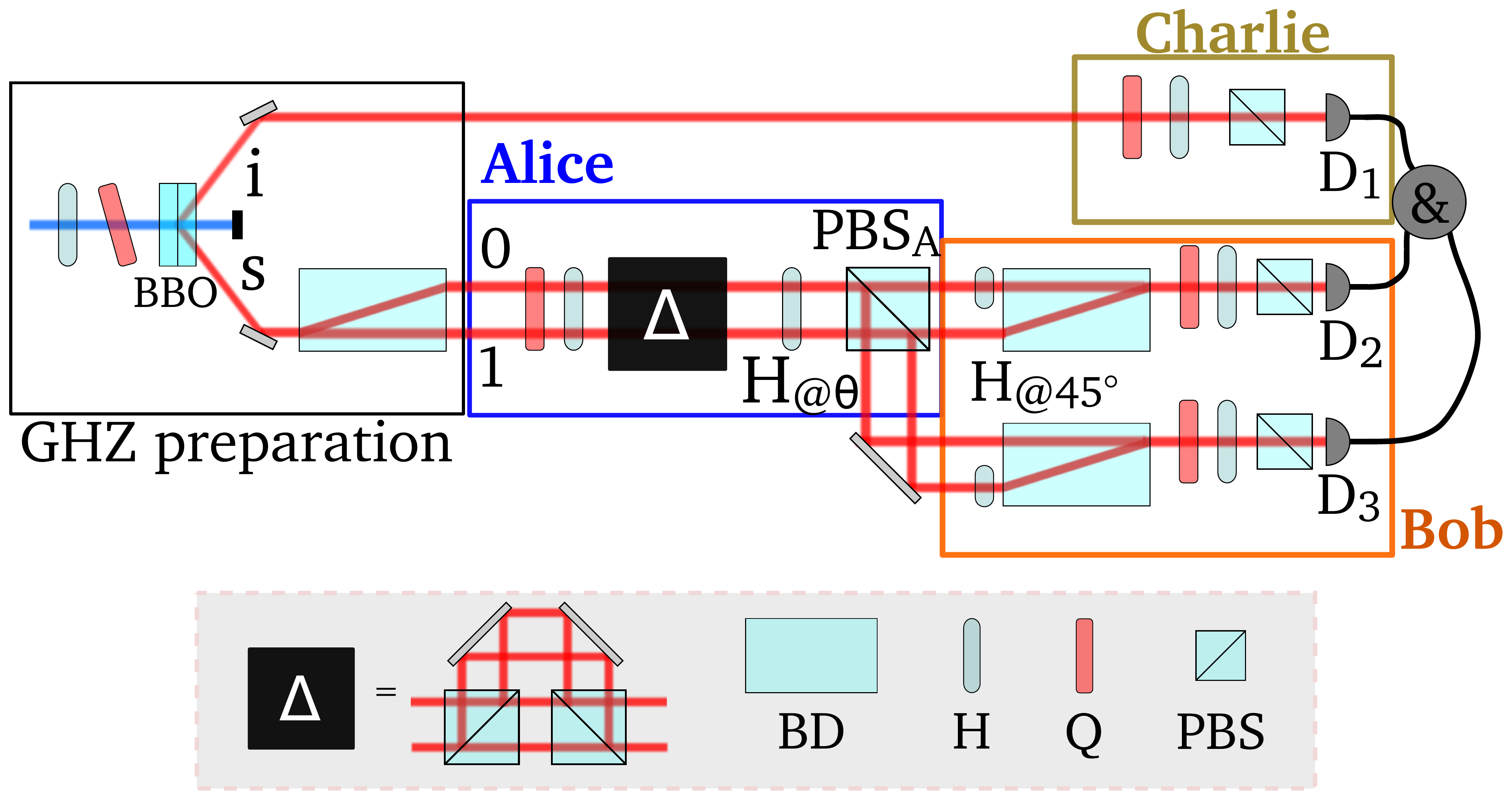}
	\caption{
Experimental setup. Two crossed-axis BBO crystals are pumped by a He-Cd laser centered at 325 nm, producing pairs of photons at 650 nm entangled in the polarization degree of freedom \cite{Kwiat1999}. The signal ($s$) photon is sent through a BD which deviates only the horizontal-polarization component, producing a tripartite GHZ state on two photons using polarization and path degrees of freedom. Idler ($i$) photons are sent directly to Charlie's polarization measurements. Signal photons are first measured in polarization by Alice, then Bob maps his path qubit onto a polarization qubit for his measurements. $H$ stands for half-wave plate, $Q$ for quarter-wave plate and $PBS$ for polarizing beam splitter. 
	}
	\label{fig:setup}
\end{figure} 

\par
To implement the wiring from Fig.\ \ref{fig:wiring}\textbf{b}, Alice's polarization measurements are realized before Bob's measurements onto the path degree of freedom. 
Alice's results are read from the output of PBS$_A$, which determines whether D$_2$ ($a=0$) or D$_3$ ($a=1$) clicks. For Alice's trivial measurement ($x=0$), crucial for the original assemblage to be LHS-decomposable, both her wave plates located before the imbalanced interferometer (represented by $\Delta$) are kept at $0^{\circ}$, and H$_{@\theta}$ is adjusted to $22.5^{\circ}$. The role of $\Delta$ is to remove the coherence between horizontal and vertical polarization components, ensuring that the photon exits PBS$_A$ randomly, independent of the input polarization state. For $x=1$, Alice's wave plates are set to project the polarization on the $X$ eigenstates,  the interferometer and H$_{@\theta}$ ($\theta=0^{\circ}$) play no role.  Bob performs his projective measurements by first mapping the path degrees of freedom onto polarization using BDs and then projecting the polarization state using his set of wave plates and  PBSs, as was realized in Ref.\  \cite{Farias2012}.  To reconstruct the assemblage in Eq.\ \eqref{eq:quantumorig}, measurements for $y=0$ and $y=1$ are made in both detectors D$_2$ and D$_3$, varying the angle of the wave plates in Bob's box. To collect the data corresponding to the wired assemblage \eqref{eq:quantumsteerable} only the $y=0$ measurement is made in D$_2$ ($a=0$) and only $y=1$ is made in D$_3$ ($a=1$), enforcing that Bob's input equals Alice's output ($y=a$).

\begin{figure}
	\centering
	\includegraphics[width=0.99\columnwidth]{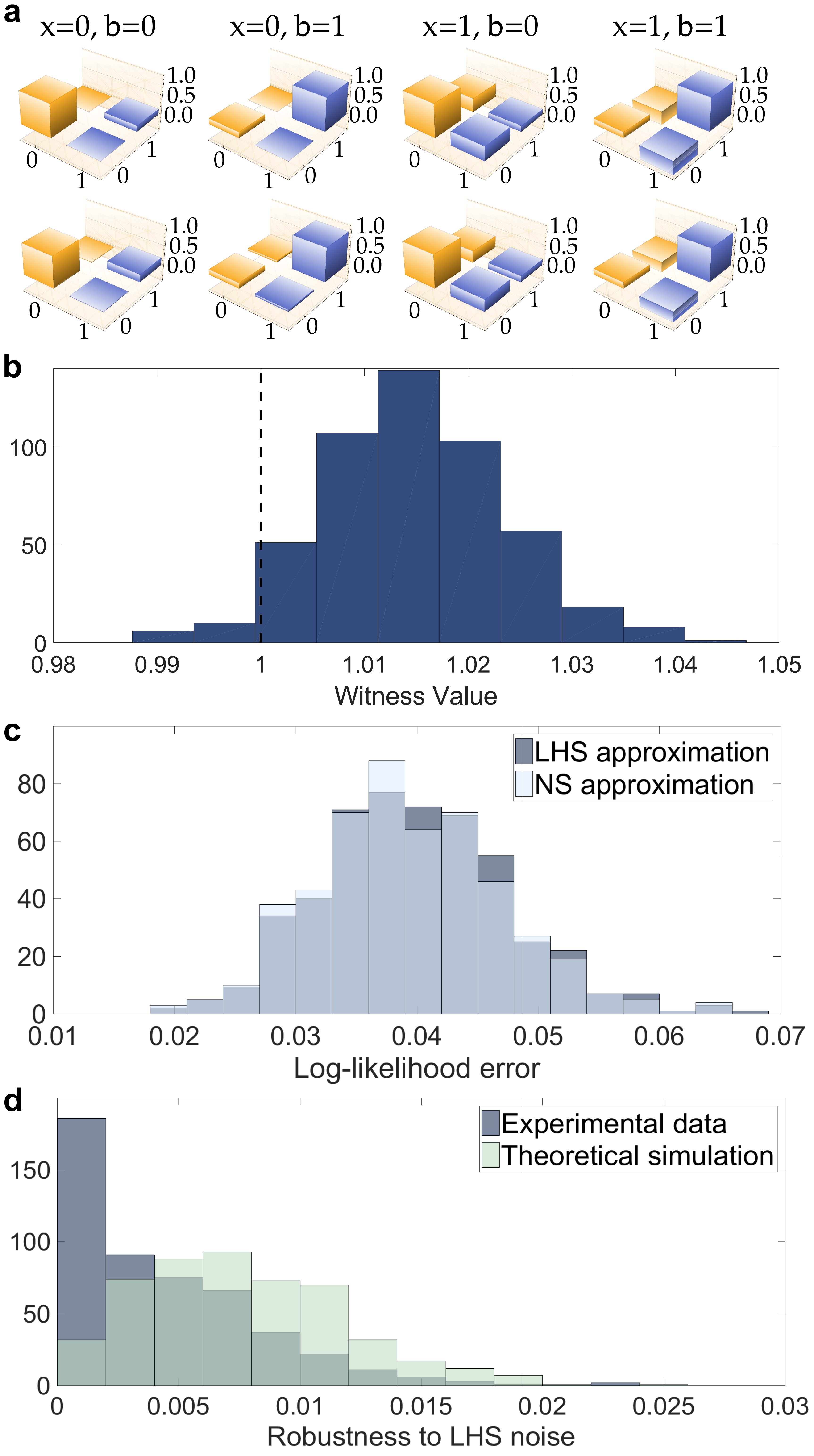}
	\caption{Experimental results. \textbf{a}, \textbf{b} Experimental assemblages after $y=a$ wiring. \textbf{a} Real part of Charlie's conditional density matrices,  theoretical (top) and experimental (bottom). \textbf{b} Steering-witness histogram. The witness value is $1.015\pm0.009$, meaning that the experimental assemblage is more than one standard deviation above the steering threshold (dashed line). \textbf{c}, \textbf{d} Compatibility of the tripartite experimental assemblage with the naive (LHS) definition of unsteerability [Eq.\ \eqref{eq:LHS}]. \textbf{c} Histogram of the error of approximating the tripartite assemblage by an NS and an LHS assemblage, showing that the error of assuming the LHS decomposition is as small as that of the physically necessary NS assumption. \textbf{d} From the best NS approximation to the experimental data, histogram of the LHS-robustness, a measure of deviations from the set \textsf{LHS}. Even with all experimental error, there is only a residual amount of robustness, fully compatible with that of the theoretical LHS assemblage solely under finite-statistics error. All histograms come from Monte Carlo simulation assuming Poisson distributions.}
	\label{fig:wiredass}
\end{figure} 

The assemblage $\boldsymbol\sigma^{(\text{GHZ})}$ was obtained experimentally by performing state tomography on Charlie's system for each measurement setting and outcome of Alice and Bob. Sixteen density matrices (plotted in Figure \ref{fig:assembtrip}, in Supplementary Material) are obtained through maximum likelihood, and the assemblage presents a fidelity-like measure of $98.2\pm0.2\%$  compared to the theoretical one (see Methods). The experimental wired assemblage is shown in Fig.\ \ref{fig:wiredass}\textbf{a}, and returns a fidelity of $98.1\pm0.6\%$ with respect to the theoretical wired assemblage given in \eqref{eq:quantumsteerable}.

An exact LHS decomposition of the experimental assemblage is not feasible due to imperfections and finite statistics --- in fact, assemblages reproducing raw experimental data exactly are not even physical, since they disobey the NS principle \cite{Cavalcanti2015a}. To show that the experimental tripartite assemblage is statistically compatible with an LHS decomposition, we proceed as follows: First, we assume the photocounts obtained for each measured projector are averages of Poisson distributions; with a Monte Carlo simulation, we sample many times each of these distributions and reconstruct the corresponding assemblages. Second, for each reconstructed assemblage, we find the physical (NS) assemblage that best approximates it through maximum-likelihood estimation, as well as the best LHS approximation for comparison. As an initial indication of LHS-compatibility, the log-likelihood error of both approximations is extremely similar, see Fig.\ \ref{fig:wiredass}\textbf{c}. Third, for the NS approximations we calculate the LHS-robustness \cite{Sainz2016a}, a measure which is zero for all LHS assemblages. For comparison, we repeat the procedure starting with simulated finite-photocount statistics from the theoretical LHS assemblage from Eq.\ \eqref{eq:quantumorig}. In Fig.\ \ref{fig:wiredass}\textbf{d} we see that the experimental robustness has a sizable zero component and  a distribution fully compatible with that of an LHS assemblage under finite measurement statistics.

To show that the experimental wired assemblage is steerable, we tested it on the optimal steering witness $W$ with respect to assemblage \eqref{eq:quantumsteerable} (see Supplementary Notes \ref{app:proof_exposure_universal_quantum}). This returned a value $1.015\pm0.009 \nleqslant 1$ (theoretical: $1.0721\nleqslant1$), where the inequality violation implies steering, see Fig.\ \ref{fig:wiredass}\textbf{b}. This allows us to conclude that the bipartite wired assemblage is indeed steerable. The experimental error was calculated using 500 assemblages also from a Monte Carlo simulation of measurement results with Poisson photocount statistics.

Using the same experimental setup, we can also experimentally demonstrate super-exposure of Bell nonlocality. As argued above, the initial experimental assemblage is compatible with an LHS model. Therefore, no matter what measurement Charlie makes, the corresponding Bell behavior will be compatible with an LHV model. Hence, we must only show that the experimental wired assemblage is Bell nonlocal. In Ref.\  \cite{Taddei2016}, a necessary and sufficient criterion for Bell nonlocality of assemblages was derived: 
Given Alice and Bob's wired measurements ($y=a$) with input bit $x$ and output bit $b$, to maximally violate a Bell inequality, Charlie performs von Neumann measurements in the $2Z+X$ and $X$ bases, labeled by input bit $z$, obtaining binary output result $c$. They thus obtain sixteen probabilities $P(b,c|x,z)$, which are used to calculate the Clauser-Horne-Shimony-Holt (CHSH) inequality \cite{Clauser1969}.
We obtained an experimental violation of $2.21\pm0.04\nleqslant2$ (theoretical prediction: $2.29\nleqslant2$), showing Bell nonlocality.

 This experiment is sufficient to for a proof-of-principle demonstration of both exposure of  steering and super-exposure of Bell nonlocality.  We note that strict demonstration of these phenomena in their appropriate DI scenarios requires a realization with space-like separation between the parties (locality loophole), as well high-efficiency source and detectors (fair-sampling assumption).

 
\section*{Discussion}\label{sec:conclusion} 
We have demonstrated that the traditional definition of multipartite steering for more than one untrusted party based on decomposability in terms of generic bilocal hidden-state models  presents inconsistencies with a widely accepted, basic notion of locality. We have also shown how, according to such definition, a broad set of steerable (exposure) and even Bell-nonlocal (super-exposure) assemblages would be created from scratch, e.g.\ by bilocal wirings acting on a seemingly unsteerable assemblage, i.e.\ an LHS one. A surprising discovery that we have made is the fact that exposure of quantum nonlocality is a universal effect, in the sense that all steering assemblages as well as Bell behaviors can be obtained as the result of an exposure protocol starting from bilocal correlations in a scenario with one more untrusted party. This highlights the power of exposure as a resource-theoretic transformation. However, we also delimit such power: we prove a no-go theorem for multi-black-box universal steering bits: there exists no single assemblage with many untrusted and one trusted party from which all assemblages with one untrusted and one trusted party can be obtained through generic free operations of steering. To restore operational consistency, we offer a redefinition of both bipartite steering in multipartite scenarios and genuinely multipartite steering that does not leave room for creating correlations from scratch.
Finally, both steering exposure and Bell nonlocality super-exposure have been demonstrated experimentally using an optical implementation. 
This is to our knowledge the first experimental observation of exposure of quantum nonlocality reported, not only in semi device-independent scenarios but also in fully device-independent ones, as originally predicted in \cite{Gallego2012,Bancal2013}.

Finally, we mention practical implications that our results might have. Steering in the scenario we work on, with a single trusted party, has been shown to be particularly relevant for the task of quantum secret sharing \cite{Xiang2017,Kogias2017}. In it, the trusted party deals a secret to the untrusted parties, who must be able to access it only when cooperating, not independently. A form of steering that is only observable when such parties cooperate, as in the exposure protocol, fits this mold quite specifically. This indicates a potential application of our results, possibly in conjunction with the open question of other joint operations able to achieve exposure.

\section*{Methods}
\subsection*{Experimental Assemblages}
Let us describe the quantum state and the assemblages produced in our experiment in more detail.
Although we treat two of the qubits as black boxes, in order to ensure that the resulting assemblage is coming up from quantum measurements performed onto a GHZ, we first made a state tomography to determine the tripartite quantum state. This can be done without adding any optical element to the setup. By varying the angles on Alice's quarter-wave plate and half-wave plate before the imbalanced interferometer, we set her apparatus to make any tomographic measurement in polarization if we set $H_{@\theta}$ to $0^{\circ}$. The tomographic projections for the path degree of freedom of photons in $s$ and polarization of photons in $i$ is done using the set of wave plates just before detectors D$_1$ and D$_2$, respectively. Using the collected coincidence counts we reconstructed the tripartite quantum state by maximum likelihood. The reconstructed density matrix is shown on Figure\,\ref{fig:ghz}. The experimental state presents  fidelity with GHZ state equals to $0.981\pm0.004$. 

\begin{figure}[tb]
	\centering
	\includegraphics[width=0.95\columnwidth]{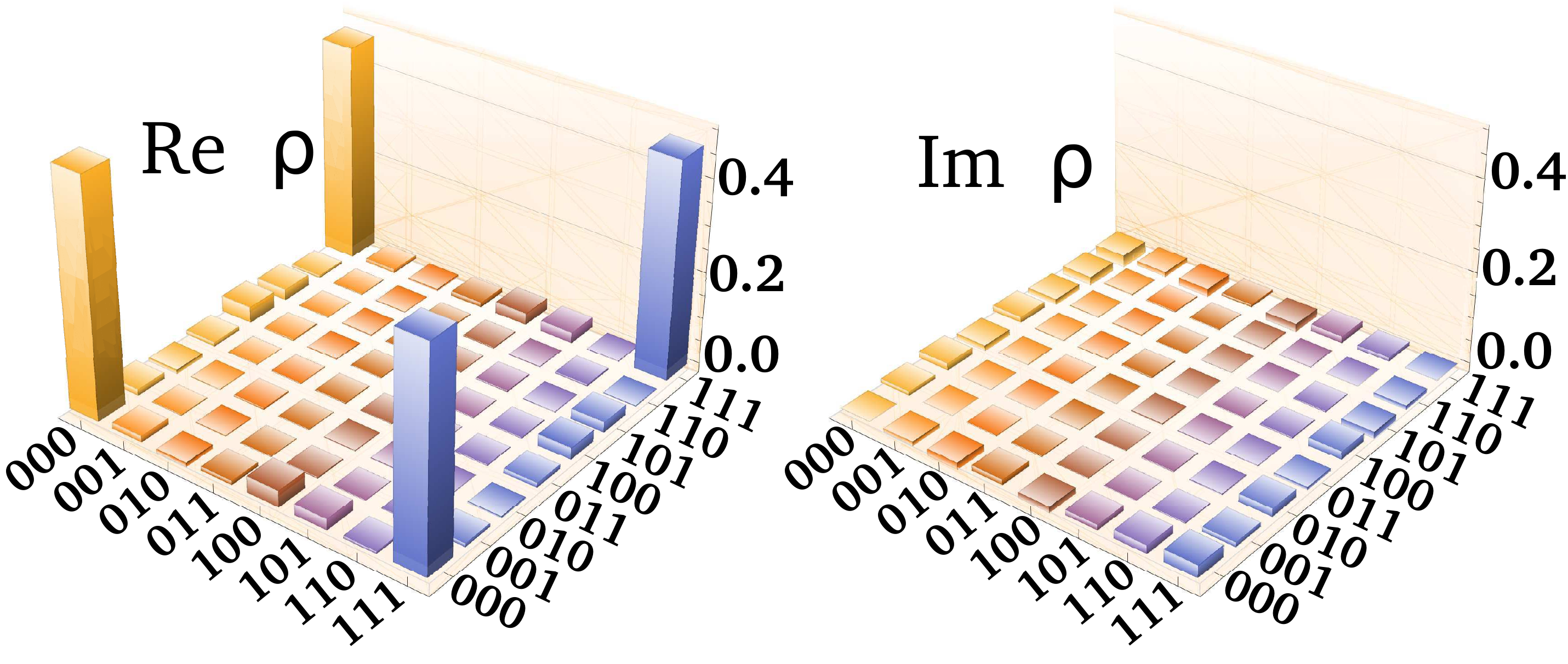}
	\caption{Real and imaginary parts of the experimental reconstructed GHZ state. Colors are for visualization purposes only. }\label{fig:ghz}
\end{figure} 

Each element of the tripartite assemblage is composed of Charlie's conditional quantum state and  the conditional probability $P_{a,b|x,y}$ for the black boxes. All sixteen experimental Charlie's density matrices are shown in Figure\,\ref{fig:assembtrip} (Supplementary Material) in comparison with the corresponding theoretical ones. The associated conditional probabilities are also shown. 

For the wired assemblage, the expected conditional probability of each outcome is $\frac12$; the experimental values are $0.46\pm0.01$, $0.54\pm0.01$, $0.49\pm0.01$, $0.51\pm0.01$ (following the order in Fig.\ \ref{fig:wiredass}\textbf{a}). The imaginary components of the density matrix average to $0.05\pm0.02$ (theoretical: zero). 

\subsection*{Assemblage Fidelity}
We can see by visual inspection that the experimental and corresponding theoretical assemblage elements shown in Figs.\ \ref{fig:wiredass} and \ref{fig:assembtrip} (Supplementary Material) are similar. To quantify this similarity we use a mean assemblage fidelity between two assemblages $\boldsymbol{\sigma}_1=\{P_1(\mathbf{a}|\mathbf{x})\varrho_1(\mathbf{a}|\mathbf{x})\}$ and $\boldsymbol\sigma_2=\{P_2(\mathbf{a}|\mathbf{x})\varrho_2(\mathbf{a}|\mathbf{x})\}$ defined by
\begin{multline}
 F(\boldsymbol\sigma_1,\boldsymbol\sigma_2)=\\
\frac{1}{N_x}\sum_{\mathbf{x},\mathbf{a}}\sqrt{P_1(\mathbf{a}|\mathbf{x})P_2(\mathbf{a}|\mathbf{x})}\mathcal{F}\left(\varrho_1(\mathbf{a}|\mathbf{x}),\varrho_2(\mathbf{a}|\mathbf{x})\right),
\end{multline}
where $\mathbf{x}$ ($\mathbf{a}$) is a list of inputs (outputs) of all black boxes, $N_x$ is the number of different measurement choices, and $\mathcal{F}(\varrho_1,\varrho_2)$ is the usual fidelity between two quantum states. The numerical values of assemblage fidelity in the main text are calculated with this definition. The above defined fidelity can be seen as a mean of the fidelities of the quantum parts weighted by the square root of blackbox probabilities. It has the property of being $1$ if all elements of the two assemblages are equal and vanishes if all quantum states are orthogonal. \\


\section*{Data and code availability}
The datasets and programming codes generated and/or analyzed during the current study are available from the corresponding author on reasonable request.


\begin{acknowledgments} 
We thank Elie Wolfe and an anonymous referee for independently pointing out a mistake in an earlier version of our manuscript. 
The authors acknowledge financial support from the Brazilian agencies CNPq (PQ grants 311416/2015-2, 304196/2018-5 and INCT-IQ), FAPERJ (PDR10 E-26/202.802/2016, JCN E-26/202.701/2018, E-26/010.002997/2014, E-26/202.7890/2017), CAPES (PROCAD2013), and the Serrapilheira Institute (grant number Serra-1709-17173).  SPW received support from Fondo Nacional de Desarrollo Cient\'{i}fico y Tecnol\'{o}gico (ANID) (1200266) and  ANID – Millennium Science Initiative Program – ICN17\_012.
\end{acknowledgments}

\section*{Competing Interests}
The authors declare no competing interests.

\section*{Author Contributions}
The theorems were derived by MMT (analytical) and RVN (codes). MMT and LA performed the causal analysis, and wrote most of the manuscript, with contributions from all authors. GHA and SPW have designed the experiment, which was performed by TLS and GHA. TLS and RVN analyzed the results. SPW and LA conceived the original idea of exploring wirings as a resource.


%
\clearpage


\title{Supplemental Material to:\\Exposure of multipartite quantum nonlocality}
\maketitle

\renewcommand{\thetable}{S\arabic{table}}\setcounter{table}{0}
\renewcommand{\thethm}{S\arabic{thm}}\setcounter{thm}{0}
\renewcommand{\thedfn}{S\arabic{dfn}}\setcounter{dfn}{0}

\onecolumngrid

This Supplementary Material is composed of Figure \ref{fig:assembtrip} and Supplementary Notes, divided into Sections \ref{app:proof_exposure_universal} through \ref{app:proof_nobits}. Citations refer to the list of references of the main paper.

\twocolumngrid

\begin{figure*}
	\centering
	\vspace{1cm}
	\includegraphics[width=.95\textwidth]{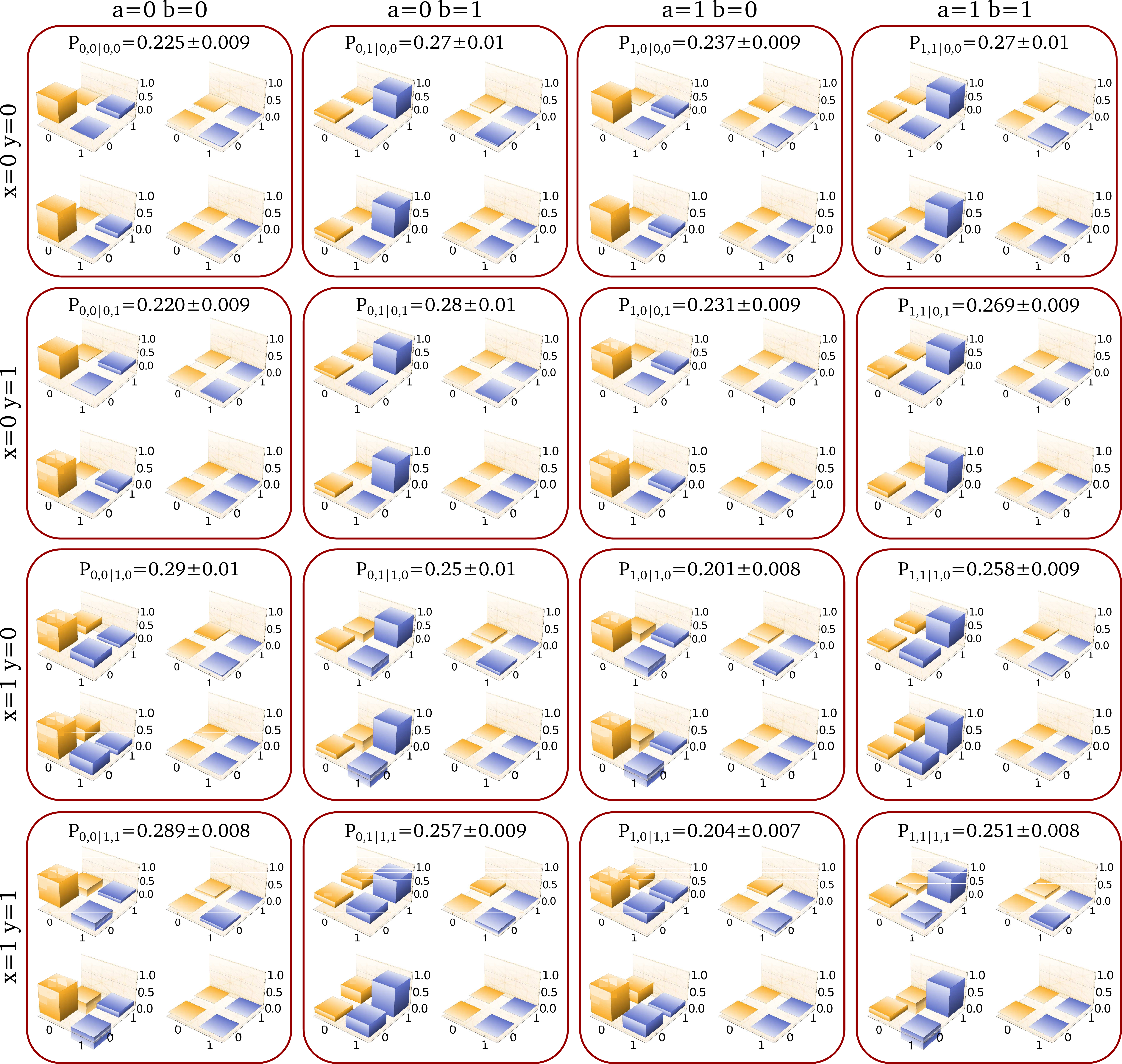}
	\vspace{.2cm}
	\caption{Theoretical and experimental reconstructed assemblages for different values of inputs $x, y$ and outputs $a, b$, as mentioned in Methods. Colors are for visualization purposes only. Each box shows the joint probability of  measurement for the black boxes, real (top left) and imaginary (top right) parts of the experimental density matrix of Charlie's partition, and real (bottom left) and imaginary (bottom right) parts of theoretical Charlie's density matrix. The theoretical probability is $0.25$ for all measurement choices and measurement outputs.}
	\vspace{.5cm}
	\label{fig:assembtrip}
\end{figure*}
\clearpage

\section*{Supplementary Notes}
\appendix{
\renewcommand\appendixname{}
\section{Universal exposure of quantum nonlocality}
\label{app:proof_exposure_universal}

In this section we prove Theorem \ref{th:universal}, i.e., that the wiring produces the desired targets and that the source assemblage $\boldsymbol\sigma^{(\text{initial})}$ and behavior $\boldsymbol{P}^{(\text{initial})}$ in Eqs.\ (\ref{eq:genactiv},\ref{eq:genblackboxactiv}) admit an LHS and an LHV models, respectively, across the bipartition $AB|C$.

\begin{proof}
It is straightforward to check that applying the wiring $y=a$ to Eqs.\ \eqref{eq:genactiv} and \eqref{eq:genblackboxactiv} of the main text, the target assemblage and behavior are obtained, i.e., $\sum_a\sigma_{a,b|x,y=a}^{(\text{initial})}=\sigma_{b|x}^{(\text{target})}$ and $\sum_aP^{(\text{initial})}(a,b,c|x,y=a,z)=P^{(\text{target})}(b,c|x,z)$.
 
We construct an explicit LHS model for the source assemblage $\boldsymbol\sigma^{(\text{initial})}$. It is given by
\begin{subequations}
\begin{align}
P_{\lambda}=\frac12 \Tr\left(\sigma_{\lambda_0|\lambda_1}^{(\text{target})}\right)  ,  \ \ &\varrho_{\lambda}=\frac{\sigma_{\lambda_0|\lambda_1}^{(\text{target})}}{\Tr\left(\sigma_{\lambda_0|\lambda_1}^{(\text{target})}\right)}, \label{eq:LHSorig1}\\
P_{a,b|x,y;\lambda}  & =\delta_{\lambda_0,b}\ \delta_{\lambda_1,x\oplus a\oplus y} \ ,
\label{eq:LHSorig2}
\end{align}\label{eq:LHSorig}\end{subequations}
where $\lambda=(\lambda_0,\lambda_1)$ is a two-bit hidden variable. 

For the Bell behavior, this expression readily lends itself for a local hidden-variable decomposition of $\boldsymbol P^{(\text{initial})}$ on $AB|C$,  $P^{(\text{initial})}(a,b,c|x,y,z) = \sum_\lambda P_{\lambda}P_{a,b|x,y;\lambda} P(c|z;\lambda)$ with the same bipartite distribution from Eq.~\eqref{eq:LHSorig2} and 
\begin{align}
P_{\lambda}=\frac{P(\lambda_0|\lambda_1)}2; \ \ \ & P(c|z;\lambda)  = \frac{P^{(\text{target})}(\lambda_0,c|\lambda_1,z)}{P(\lambda_0|\lambda_1)} \ ,\label{eq:LHVorig1}
\end{align}
where $P(\lambda_0|\lambda_1):=\sum_cP^{(\text{target})}(\lambda_0,c|\lambda_1,z)$.
\end{proof}

\section{Quantum-realizable exposure of quantum nonlocality}
\label{app:proof_exposure_universal_quantum}
In this section we prove Theorem \ref{th:quantum_protocol}, i.e., that the physically-realizable source assemblage $\boldsymbol\sigma^{(\text{GHZ})}$ in Eq.\ \eqref{eq:quantumorig} admits an LHS model across the bipartition $AB|C$, that the resulting wired assemblage is that of Eq.\ \eqref{eq:quantumsteerable}, and that the latter is both steerable and Bell nonlocal.
\begin{proof}
The LHS decomposition for Eq.\ \eqref{eq:quantumorig} is found via semi-definite programming (SDP). SDP is a convex optimization procedure for linear objective functions that is particularly useful in the semi-DI scenario \cite{Cavalcanti2017}. The numerical results in this case allow one to find analytic formulas for the decomposition, namely 
\begin{subequations}\begin{align}
P_{\lambda}= \frac14 ;\ \  \ \varrho_{\lambda} &= \frac{\mathbb1}2 + \frac{(-1)^{\lambda_0}}{2\sqrt2} \left[ Z+(-1)^{\lambda_1}X\right] ; \label{eq:quantumLHSdecomp1}\\
P_{a,b|x,y;\lambda} &=  \delta_{\lambda_0,b} \frac{1+x(-1)^{a+y+\lambda_1}}2 \  , \label{eq:quantumLHSdecomp2}
\end{align}\label{eq:quantumLHSdecomp}\end{subequations}
where again $\lambda=(\lambda_0,\lambda_1)$ is a two-bit hidden variable. 

Let us now prove the steerability and Bell-nonlocality of assemblage \eqref{eq:quantumsteerable}.
Steerability: with an SDP, we have obtained an assemblage-like object $W=\{w_{a|x}\}_{a,x}$ that serves as a steering witness, i.e.\ it establishes the inequality $\sum_{a,x}\Tr\left[w_{a|x}\sigma_{a|x}\right]\leqslant1$, which can only be violated if assemblage $\boldsymbol\sigma=\{\sigma_{a|x}\}_{a,x}$ is steerable.
Optimized for assemblage \eqref{eq:quantumsteerable}, the witness returns a value of $1.0721$ and can be cast as
\begin{equation}
w_{0|0} = \begin{bmatrix}
p & -c \cr -c & 1-p
\end{bmatrix},\quad w_{0|1} = \begin{bmatrix}
q & p/2 \cr p/2 & -q
\end{bmatrix},
\label{eq:witnessW}\end{equation}
with $p =\frac{1}{2}(1+\frac{1}{\sqrt{5}}),\,c \approx 0.1382,\,q \approx 0.2236$, and $w_{1|x} = Y\,w_{0|x}\,Y,\,x=0,1$.
Bell-nonlocality: The necessary and sufficient criterion from \cite{Taddei2016} yields an optimal violation of the Clauser-Horne-Shimony-Holt (CHSH) inequality of $|-\frac{\sqrt5+1}{\sqrt2}|\approx2.29\nleqslant2$, attained when Charlie makes von Neumann measurements in the eigenbases of $2Z+X$ and $X$. 
\end{proof}

\section{On the sets of LHS assemblages, TO-LHS assemblages, and NS-LHS assemblages}
\label{app:strict_subset}

\begin{table*}[t]%
\setlength{\tabcolsep}{0.5em} 
\normalsize
\begin{tabular}{cc|cc|c}
$a$&$b$&$x$&$y$&$\sigma_{a,b|x,y}^{W}$ \\ \hline
0&0&0&0& $\frac16\left[2\eta^2\ket0\bra0+(1+\sqrt{1-\eta^2}-\eta^2/2)\ket1\bra1+\eta(1+\sqrt{1-\eta^2})X\right]$ \\ 
0&1&0&0& $\frac16\left[2(1-\eta^2)\ket0\bra0+\eta^2/2\ket1\bra1-\eta\sqrt{1-\eta^2}X\right]$ \\ 
1&0&0&0& $\frac16\left[2(1-\eta^2)\ket0\bra0+\eta^2/2\ket1\bra1-\eta\sqrt{1-\eta^2}X\right]$ \\ 
1&1&0&0& $\frac16\left[2\eta^2\ket0\bra0+(1-\sqrt{1-\eta^2}-\eta^2/2)\ket1\bra1-\eta(1-\sqrt{1-\eta^2})X\right]$ \\ \hline
0&0&0&1& $\frac1{12}\left[2(1+2\eta\sqrt{1-\eta^2})\ket0\bra0+(1-\eta+\sqrt{1-\eta^2}-\eta\sqrt{1-\eta^2})\ket1\bra1+(1+\eta+\sqrt{1-\eta^2}-2\eta^2)X\right]$ \\ 
0&1&0&1& $\frac1{12}\left[2(1-2\eta\sqrt{1-\eta^2})\ket0\bra0+(1+\eta+\sqrt{1-\eta^2}+\eta\sqrt{1-\eta^2})\ket1\bra1-(1-\eta+\sqrt{1-\eta^2}-2\eta^2)X\right]$ \\ 
1&0&0&1& $\frac1{12}\left[2(1-2\eta\sqrt{1-\eta^2})\ket0\bra0+(1-\eta-\sqrt{1-\eta^2}+\eta\sqrt{1-\eta^2})\ket1\bra1-(1+\eta-\sqrt{1-\eta^2}-2\eta^2)X\right]$ \\ 
1&1&0&1& $\frac1{12}\left[2(1+2\eta\sqrt{1-\eta^2})\ket0\bra0+(1+\eta-\sqrt{1-\eta^2}-\eta\sqrt{1-\eta^2})\ket1\bra1+(1-\eta-\sqrt{1-\eta^2}-2\eta^2)X\right]$ \\ \hline
$a$&$b$&1&0& $\sigma_{a,b|1,0}^{W}=\sigma_{b,a|0,1}^{W}$\\ \hline
0&0&1&1& $\frac16\left[2(1-\eta^2)\ket0\bra0+(1-\eta-(1-\eta^2)/2)\ket1\bra1+\sqrt{1-\eta^2}(1-\eta)X\right]$ \\ 
0&1&1&1& $\frac16\left[2\eta^2\ket0\bra0+(1-\eta^2)/2\ket1\bra1+\eta\sqrt{1-\eta^2}X\right]$ \\ 
1&0&1&1& $\frac16\left[2\eta^2\ket0\bra0+(1-\eta^2)/2\ket1\bra1+\eta\sqrt{1-\eta^2}X\right]$ \\ 
1&1&1&1& $\frac16\left[2(1-\eta^2)\ket0\bra0+(1+\eta-(1-\eta^2)/2)\ket1\bra1-\sqrt{1-\eta^2}(1+\eta)X\right]$ \\ 
\end{tabular}
\caption{Example quantum assemblage to demonstrate strict inclusion of \textsf{NS-LHS} in \textsf{TO-LHS}.}
\label{tab:Wassemblage}
\end{table*}

We now state a theorem that sustains Fig.\ \ref{fig:setsLHS}, concerning the inclusion relations between the sets \textsf{Q-LHS}, \textsf{NS-LHS}, \textsf{TO-LHS}, and \textsf{LHS}.
\begin{thm}
\textsf{Q-LHS} $\subset$ \textsf{NS-LHS} $\subset$ \textsf{TO-LHS} $\subset$ \textsf{LHS}, and these relations also hold strictly if we restrict to quantum-realizable assemblages.\label{th:sets_LHS}\end{thm} 

\begin{proof} 
From the definitions in Eqs.\ (\ref{eq:LHS},\ref{eq:TOdefinition}), it is clear that \textsf{Q-LHS} $\subseteq$ \textsf{NS-LHS} $\subseteq$ \textsf{TO-LHS} $\subseteq$ \textsf{LHS}. The phenomenon of exposure implies that the assemblages in Eqs.\ (\ref{eq:genactiv},\ref{eq:quantumorig}) belong to \textsf{LHS}, but not to \textsf{TO-LHS}, so the inclusion of one in the other is strict (notice that assemblage \eqref{eq:quantumorig} is quantum realizable). In \cite{Cavalcanti2015a}, assemblages that violate the quantum-LHS model are found from tripartite quantum states under independent measurements on Alice and Bob, hence must admit an NS-LHS model, which shows that \textsf{Q-LHS} $\subset$ \textsf{NS-LHS} strictly. 
To prove that \textsf{NS-LHS} is a strict subset of \textsf{TO-LHS}, we need an example of a TO-LHS assemblage that does not belong to \textsf{NS-LHS}. 

One way to do so is to follow the reasoning of \cite{Gallego2012}: take the time-ordered decomposition of the distribution $\boldsymbol P$ from \cite{Gallego2011} that violates the guess-your-neighbor's-input (GYNI) inequality and find the $\varrho_{\lambda}$ that best mimic the marginal $P_{a|x,\lambda}$ --- this effectively amounts to a one-time program \cite{Roehsner2018}. The resulting TO-LHS assemblage violates GYNI, hence is not NS-LHS, but it is also supra-quantum, since no quantum state can violate the GYNI inequality.

To find a quantum-realizable assemblage that belongs to \textsf{TO-LHS}, but not to \textsf{NS-LHS}, we take inspiration from Bancal \emph{et al} \cite{Bancal2013}, who have found Bell behaviors obtainable from noisy $W$ states with the analogous DI-scenario property (TO-LHV, but not NS-LHV). A pure $W$ state is given by $\ket W:=(\ket{001}+\ket{010}+\ket{100})/\sqrt3$, its noisy version with visibility $v$, by 
\begin{equation}
\rho_W= v \ \ket W\bra W+ (1-v)\ \mathbb1^{(ABC)}/8 \ .
\label{eq:noisyWstate}
\end{equation}
Alice and Bob make von Neumann measurements on the bases $\eta X+ \sqrt{1-\eta^2}Z$ ($x$ or $y=0$) and $\sqrt{1-\eta^2}X-\eta Z$ ($x$ or $y=1$), with $\eta\approx0.97177$, which yields the assemblage  
\begin{equation}
\sigma_{a,b|x,y}^{\text{noisy } W} = v \ \sigma_{a,b|x,y}^{W} + (1-v)\ \mathbb1^{C}/8 \ ,
\label{eq:noisyWassemb}
\end{equation}
where $\sigma_{a,b|x,y}^{W}$ is given in Table \ref{tab:Wassemblage}.
These measurements, together with an appropriate measurement by Charlie, yield in \cite{Bancal2013} a DI-inequality violation requiring minimal visibility.

We obtain the optimal NS-LHS witness $\boldsymbol{W}=\{W_{abxy}\}_{a,b,x,y}$ for $\sigma_{a,b|x,y}^{\text{noisy } W}$ for $v=0.58$, i.e. $\boldsymbol{W}$ satisfies the property
\begin{equation}
-1 \leq \sum_{a,b,x,y}\Tr[W_{abxy}\, \sigma^{\text{NS-LHS}}_{a,b\vert x,y}]\leq 0
\label{eq:witness_bound}
\end{equation}
for every NS-LHS assemblage $\boldsymbol\sigma^{\text{NS-LHS}}$. Its components $W_{abxy}$ are given in Table \ref{tab:NS-LHSwitness}. This witness is violated by $\sigma_{a,b|x,y}^{\text{noisy } W}$ from $v\approx0.58$ onwards; for $v=0.64$, it returns 0.0301.

\begin{table*}[htb]%
\setlength{\tabcolsep}{0.5em} 
\normalsize
\begin{tabular}{|c|cccc|}\hline
\diagbox{$x,y$}{$a,b$} &$00$&$01$&$10$&$11$\\ \hline 
&&&&\\[-2ex]
$00$ &
$\begin{bmatrix}-0.0056 & \mspc 0.1194 \\ \mspc 0.1194 & -0.1205	\end{bmatrix}$&
$\begin{bmatrix}-0.1394 & -0.0603\\ -0.0603 & \mspc 0.0662				\end{bmatrix}$&
$\begin{bmatrix}-0.1394 & -0.0603\\ -0.0603 & \mspc 0.0662				\end{bmatrix}$&
$\begin{bmatrix}\mspc 0.0239 & -0.0656 \\ -0.0656 & -0.1869				\end{bmatrix}$\\[2ex]
$01$ &
$\begin{bmatrix}\mspc 0.0233 & -0.0324 \\ -0.0324 & -0.1706				\end{bmatrix}$&
$\begin{bmatrix}-0.2194 & \mspc 0.1346 \\ \mspc 0.1346 & -0.0079	\end{bmatrix}$&
$\begin{bmatrix}-0.0560 &\mspc 0.1109\\ \mspc 0.1109 &\mspc 0.0114\end{bmatrix}$&
$\begin{bmatrix}-0.0417 & -0.1490 \\ -0.1490 & -0.1079						\end{bmatrix}$\\[2ex]
$10$ &
$\begin{bmatrix}\mspc 0.0233 & -0.0324 \\-0.0324 & -0.1706				\end{bmatrix}$&
$\begin{bmatrix}-0.0560 &\mspc 0.1109 \\ \mspc 0.1109&\mspc 0.0114\end{bmatrix}$&
$\begin{bmatrix}-0.2194 &\mspc 0.1346 \\ \mspc 0.1346 & -0.0079		\end{bmatrix}$&
$\begin{bmatrix}-0.0417 & -0.1490 \\ -0.1490 & -0.1079						\end{bmatrix}$ \\[2ex]
$11$&
$\begin{bmatrix}-0.0410 & -0.0560 \\ -0.0560 & \mspc 0.0863				\end{bmatrix}$&
$\begin{bmatrix}\mspc 0.0665 &\mspc 0.0431\\ \mspc 0.0431 &-0.2194\end{bmatrix}$&
$\begin{bmatrix}\mspc 0.0665 &\mspc 0.0431\\ \mspc 0.0431 &-0.2194\end{bmatrix}$&
$\begin{bmatrix}	-0.4431 &  -0.0727 \\   -0.0727 &  \mspc 0.0239	\end{bmatrix}$ \\[2ex] \hline
\end{tabular}
\caption{Elements of witness $W_{abxy}$ used to demonstrate strict inclusion of \textsf{NS-LHS} in \textsf{TO-LHS}.}
\label{tab:NS-LHSwitness}
\end{table*}

\begin{table*}[htb]%
\begin{tabular}{|cc||cc||cc||cc|}\hline
\ \ \ $\lambda$ \ \ \ & $\sigma_\lambda$&\ \ \ $\lambda$\ \ \ &$\sigma_\lambda$&\ \ \ $\lambda$\ \ \ &$\sigma_\lambda$&\ \ \ $\lambda$\ \ \ & $\sigma_\lambda$ \rule{0pt}{2.5ex} \\[.5ex] \hline
{0}&$\begin{bmatrix} 	\mspc0.0045 & \mspc0.0013\\	\mspc0.0013 & \mspc0.0009\\\end{bmatrix}$&
{1}&$\begin{bmatrix}	\mspc0.0928 & \mspc0.0246\\	\mspc0.0246 & \mspc0.0070\\\end{bmatrix}$&
{2}&$\begin{bmatrix}	\mspc0.0036 & \mspc0.0011\\	\mspc0.0011 & \mspc0.0009\\\end{bmatrix}$&
{3}&$\begin{bmatrix}	\mspc0.0244 & \mspc0.0068\\	\mspc0.0068 & \mspc0.0024\\\end{bmatrix}$\rule{0pt}{4ex}\\[2ex]
\hline
{4}&$\begin{bmatrix}	\mspc0.0055 & \mspc0.0058\\	\mspc0.0058 & \mspc0.0071\\\end{bmatrix}$&
{5}&$\begin{bmatrix}	\mspc0.0084 & \mspc0.0071\\	\mspc0.0071 & \mspc0.0067\\\end{bmatrix}$&
{6}&$\begin{bmatrix}	\mspc0.0066 & \mspc0.0076\\	\mspc0.0076 & \mspc0.0098\\\end{bmatrix}$&
{7}&$\begin{bmatrix}	\mspc0.0100 & \mspc0.0090\\	\mspc0.0090 & \mspc0.0089\\\end{bmatrix}$\rule{0pt}{4ex}\\[2ex] 
\hline
{8}&$\begin{bmatrix}	\mspc0.0048 & -0.0029\\	-0.0029 & \mspc0.0025\\\end{bmatrix}$&
{9}&$\begin{bmatrix}	\mspc0.0118 & -0.0052\\	-0.0052 & \mspc0.0029\\\end{bmatrix}$&
{10}&$\begin{bmatrix}	\mspc0.0040 & -0.0026\\	-0.0026 & \mspc0.0024\\\end{bmatrix}$&
{11}&$\begin{bmatrix}	\mspc0.0079 & -0.0037\\	-0.0037 & \mspc0.0024\\\end{bmatrix}$\rule{0pt}{4ex}\\[2 ex]
\hline
{12}&$\begin{bmatrix}	\mspc0.0007 & -0.0004\\	-0.0004 & \mspc0.0024\\\end{bmatrix}$&
{13}&$\begin{bmatrix}	\mspc0.0008 & -0.0002\\	-0.0002 & \mspc0.0014\\\end{bmatrix}$&
{14}&$\begin{bmatrix}	\mspc0.0006 & -0.0004\\	-0.0004 & \mspc0.0029\\\end{bmatrix}$&
{15}&$\begin{bmatrix}	\mspc0.0007 & -0.0002\\	-0.0002 & \mspc0.0015\\\end{bmatrix}$\rule{0pt}{4ex}\\[2 ex]
\hline
{16}&$\begin{bmatrix}	\mspc0.0219 & \mspc0.0118\\	\mspc0.0118 & \mspc0.0064\\\end{bmatrix}$&
{17}&$\begin{bmatrix}	\mspc0.0001 & \mspc0.0002\\	\mspc0.0002 & \mspc0.0010\\\end{bmatrix}$&
{18}&$\begin{bmatrix}	\mspc0.0028 & -0.0005\\	-0.0005 & \mspc0.0001\\\end{bmatrix}$&
{19}&$\begin{bmatrix}	\mspc0.0002 & -0.0002\\	-0.0002 & \mspc0.0004\\\end{bmatrix}$\rule{0pt}{4ex}\\[2 ex]
\hline
{20}&$\begin{bmatrix}	\mspc0.0612 & \mspc0.0411\\	\mspc0.0411 & \mspc0.0277\\\end{bmatrix}$&
{21}&$\begin{bmatrix}	\mspc0.0034 & \mspc0.0126\\	\mspc0.0126 & \mspc0.0467\\\end{bmatrix}$&
{22}&$\begin{bmatrix}	\mspc0.0007 & -0.0001\\	-0.0001 & \mspc0.0001\\\end{bmatrix}$&
{23}&$\begin{bmatrix}	\mspc0.0002 & -0.0002\\	-0.0002 & \mspc0.0004\\\end{bmatrix}$\rule{0pt}{4ex}\\[2 ex]
\hline
{24}&$\begin{bmatrix}	\mspc0.0007 & \mspc0.0003\\	\mspc0.0003 & \mspc0.0002\\\end{bmatrix}$&
{25}&$\begin{bmatrix}	\mspc0.0001 & \mspc0.0001\\	\mspc0.0001 & \mspc0.0010\\\end{bmatrix}$&
{26}&$\begin{bmatrix}	\mspc0.0135 & -0.0036\\	-0.0036 & \mspc0.0010\\\end{bmatrix}$&
{27}&$\begin{bmatrix}	\mspc0.0074 & -0.0106\\	-0.0106 & \mspc0.0153\\\end{bmatrix}$\rule{0pt}{4ex}\\[2 ex]
\hline
{28}&$\begin{bmatrix}	\mspc0.0006 & \mspc0.0003\\	\mspc0.0003 & \mspc0.0003\\\end{bmatrix}$&
{29}&$\begin{bmatrix}	\mspc0.0010 & \mspc0.0073\\	\mspc0.0073 & \mspc0.0545\\\end{bmatrix}$&
{30}&$\begin{bmatrix}	\mspc0.0008 & -0.0002\\	-0.0002 & \mspc0.0001\\\end{bmatrix}$&
{31}&$\begin{bmatrix}	\mspc0.0015 & -0.0025\\	-0.0025 & \mspc0.0045\\\end{bmatrix}$\rule{0pt}{4ex}\\[2 ex]
\hline
{32}&$\begin{bmatrix}	\mspc0.0020 & \mspc0.0006\\	\mspc0.0006 & \mspc0.0016\\\end{bmatrix}$&
{33}&$\begin{bmatrix}	\mspc0.0049 & \mspc0.0013\\	\mspc0.0013 & \mspc0.0013\\\end{bmatrix}$&
{34}&$\begin{bmatrix}	\mspc0.0017 & \mspc0.0006\\	\mspc0.0006 & \mspc0.0018\\\end{bmatrix}$&
{35}&$\begin{bmatrix}	\mspc0.0038 & \mspc0.0011\\	\mspc0.0011 & \mspc0.0014\\\end{bmatrix}$\rule{0pt}{4ex}\\[2 ex]
\hline
{36}&$\begin{bmatrix}	\mspc0.0020 & -0.0013\\	-0.0013 & \mspc0.0022\\\end{bmatrix}$&
{37}&$\begin{bmatrix}	\mspc0.0031 & -0.0012\\	-0.0012 & \mspc0.0014\\\end{bmatrix}$&
{38}&$\begin{bmatrix}	\mspc0.0018 & -0.0013\\	-0.0013 & \mspc0.0024\\\end{bmatrix}$&
{39}&$\begin{bmatrix}	\mspc0.0026 & -0.0011\\	-0.0011 & \mspc0.0015\\\end{bmatrix}$\rule{0pt}{4ex}\\[2 ex]
\hline
{40}&$\begin{bmatrix}	\mspc0.0037 & -0.0000\\	-0.0000 & \mspc0.0009\\\end{bmatrix}$&
{41}&$\begin{bmatrix}	\mspc0.0261 & \mspc0.0009\\	\mspc0.0009 & \mspc0.0007\\\end{bmatrix}$&
{42}&$\begin{bmatrix}	\mspc0.0029 & -0.0000\\	-0.0000 & \mspc0.0010\\\end{bmatrix}$&
{43}&$\begin{bmatrix}	\mspc0.0125 & \mspc0.0005\\	\mspc0.0005 & \mspc0.0008\\\end{bmatrix}$\rule{0pt}{4ex}\\[2 ex]
\hline
{44}&$\begin{bmatrix}	\mspc0.0069 & -0.0040\\	-0.0040 & \mspc0.0032\\\end{bmatrix}$&
{45}&$\begin{bmatrix}	\mspc0.0227 & -0.0094\\	-0.0094 & \mspc0.0045\\\end{bmatrix}$&
{46}&$\begin{bmatrix}	\mspc0.0055 & -0.0034\\	-0.0034 & \mspc0.0030\\\end{bmatrix}$&
{47}&$\begin{bmatrix}	\mspc0.0140 & -0.0060\\	-0.0060 & \mspc0.0033\\\end{bmatrix}$\rule{0pt}{4ex}\\[2 ex]
\hline
{48}&$\begin{bmatrix}	\mspc0.0062 & \mspc0.0036\\	\mspc0.0036 & \mspc0.0022\\\end{bmatrix}$&
{49}&$\begin{bmatrix}	\mspc0.0011 & \mspc0.0051\\	\mspc0.0051 & \mspc0.0258\\\end{bmatrix}$&
{50}&$\begin{bmatrix}	\mspc0.0031 & -0.0006\\	-0.0006 & \mspc0.0002\\\end{bmatrix}$&
{51}&$\begin{bmatrix}	\mspc0.0007 & -0.0011\\	-0.0011 & \mspc0.0018\\\end{bmatrix}$\rule{0pt}{4ex}\\[2 ex]
\hline
{52}&$\begin{bmatrix}	\mspc0.0009 & \mspc0.0005\\	\mspc0.0005 & \mspc0.0003\\\end{bmatrix}$&
{53}&$\begin{bmatrix}	\mspc0.0001 & \mspc0.0005\\	\mspc0.0005 & \mspc0.0034\\\end{bmatrix}$&
{54}&$\begin{bmatrix}	\mspc0.0035 & -0.0008\\	-0.0008 & \mspc0.0003\\\end{bmatrix}$&
{55}&$\begin{bmatrix}	\mspc0.0193 & -0.0303\\	-0.0303 & \mspc0.0479\\\end{bmatrix}$\rule{0pt}{4ex}\\[2 ex]
\hline
{56}&$\begin{bmatrix}	\mspc0.0044 & \mspc0.0023\\	\mspc0.0023 & \mspc0.0013\\\end{bmatrix}$&
{57}&$\begin{bmatrix}	\mspc0.0002 & \mspc0.0004\\	\mspc0.0004 & \mspc0.0024\\\end{bmatrix}$&
{58}&$\begin{bmatrix}	\mspc0.0287 & -0.0055\\	-0.0055 & \mspc0.0011\\\end{bmatrix}$&
{59}&$\begin{bmatrix}	\mspc0.0008 & -0.0011\\	-0.0011 & \mspc0.0018\\\end{bmatrix}$\rule{0pt}{4ex}\\[2 ex]
\hline
{60}&$\begin{bmatrix}	\mspc0.0008 & \mspc0.0004\\	\mspc0.0004 & \mspc0.0003\\\end{bmatrix}$&
{61}&$\begin{bmatrix}	\mspc0.0001 & \mspc0.0002\\	\mspc0.0002 & \mspc0.0015\\\end{bmatrix}$&
{62}&$\begin{bmatrix}	\mspc0.0967 & -0.0246\\	-0.0246 & \mspc0.0063\\\end{bmatrix}$&
{63}&$\begin{bmatrix}	\mspc0.0206 & -0.0300\\	-0.0300 & \mspc0.0440\\\end{bmatrix}$\rule{0pt}{4ex}\\[2 ex]
\hline
\end{tabular}
\caption{Non-normalized states $\sigma_\lambda$ needed in Eq.\ \eqref{eq:DeterministicTOLHS} for the TO-LHS decomposition of the assemblage \eqref{eq:noisyWassemb}.}
\label{tab:sigmalambda}
\end{table*}

However, there is a TO-LHS decomposition of $\sigma_{a,b|x,y}^{\text{noisy } W}$ for $v=0.64$ (hence for $v<0.64$), which, equivalently to Eq.\ \eqref{eq:TOdefinition}, can be written as
\begin{subequations}
\begin{align}
\sigma_{a,b|x,y}^{\text{noisy W}}
&= \sum_{\lambda} D_\lambda(a|x) D_\lambda(b|x,y)\,\sigma_{\lambda} \label{eq:DeterministicTOLHS1}\\
&= \sum_{\lambda} D_\lambda(a|x,y) D_\lambda(b|y) \,\sigma_{\lambda} \ ,
\label{eq:DeterministicTOLHS2}
\end{align}\label{eq:DeterministicTOLHS}\end{subequations}
where the $D_\lambda$ are deterministic response functions and $\sigma_{\lambda}:=p_\lambda\rho_\lambda$ are non-normalized states. Each $D_\lambda(a|x)$ is specified by $a_x$, the deterministic outcome $a$ conditioned on $x$; the notation follows analogously for $D_\lambda(b|x,y)$, $D_\lambda(a|x,y)$, and $D_\lambda(b|y)$ ($b_{xy}$, $a_{xy}$, and $b_y$, respectively). These are given by
\begin{equation*}
\begin{tabular}{|c|cccccc|}\hline
$\lambda$ & $a_0$ & $a_1$ & $b_{00}$ & $b_{01}$ & $b_{10}$ & $b_{11}$ \\ \hline
0&0&0&0&0&0&0 \\
1&0&0&0&0&0&1 \\
2&0&0&0&0&1&0 \\
3&0&0&0&0&1&1 \\
\multicolumn{7}{c}{$\vdots$}\\
62&1&1&1&1&1&0 \\
63&1&1&1&1&1&1 \\ \hline
\end{tabular}
 \ \ \ \ 
\begin{tabular}{|c|cccccc|}\hline
$\lambda$ & $a_{00}$ & $a_{01}$ & $a_{10}$ & $a_{11}$ & $b_{0}$ & $b_{1}$ \\ \hline
0&0&0&0&0&0&0 \\
1&0&0&0&0&0&1 \\
2&0&0&0&0&1&0 \\
3&0&0&0&0&1&1 \\
\multicolumn{7}{c}{$\vdots$}\\
62&1&1&1&1&1&0 \\
63&1&1&1&1&1&1 \\ \hline
\end{tabular} \ ,
\end{equation*}
where in each table, the six columns to the right are the binary expression of the leftmost column ($\lambda$). 
The states $\sigma_{\lambda}$ are given in Table \ref{tab:sigmalambda}.
\end{proof}

As mentioned in the main text, this example implies that entanglement certification does not coincide with steering. Consider the general form of a quantum state separable in the $AB|C$ partition,
\begin{equation}
\sum_\lambda P_\lambda \ \varrho_\lambda^{AB}\otimes \varrho_\lambda^C \ ,
\label{eq:separableABtoC}
\end{equation}
where $\varrho_\lambda^{AB}$ is the $\lambda$-th hidden state for $AB$ (possibly entangled). Any local measurements made on $AB$ yield an assemblage of the form of Eq.~\eqref{eq:LHS} with a quantum-mechanical $P_{a,b|x,y,\lambda}$. That assemblage would then have a quantum-LHS (hence an NS-LHS) decomposition. Since we have seen that $\sigma_{a,b|x,y}^{\text{noisy W}}$ with $v$ in the appropriate range cannot be NS-LHS-decomposed (much less quantum-LHS decomposed), we conclude that $\sigma_{a,b|x,y}^{\text{noisy W}}$ cannot be produced from local measurements on a quantum state separable on $AB|C$, i.e.\ we certify $AB|C$ entanglement with this assemblage, which is nevertheless unsteerable.

\section{Redefinition of genuinely multipartite steering}
\label{sec:def_gen_multipartite}
Although our discussion has focused on steering along a fixed bipartition, it has a bearing on genuine multipartite steering as well. This concept hinges on bi-separability over all possible bipartitions, as used by D. Cavalcanti \emph{et al} to define genuine multipartite steering in \cite{Cavalcanti2015a}. Interestingly, however, our results can be used to generalize that definition.

\begin{dfn}[Redefinition of genuinely multipartite steering]\label{def:genuine}
An assemblage $\boldsymbol\sigma$ is genuinely multipartite steerable if it does \emph{not} admit a decomposition of the form
\begin{subequations}\begin{eqnarray}
\sigma_{a,b|x,y}	=  \sum_\mu p^{A|BC}_\mu & P_{a|x,\mu} & \sigma_{b|y}^{C}(\mu)  \label{eq:biseparableassemblage1} \\
	+  \sum_\nu			p^{B|AC}_\nu 		  & P_{b|y,\nu}  		& 		\sigma_{a|x}^{C}(\nu)  \label{eq:biseparableassemblage2} \\
	+  \sum_\lambda p^{AB|C}_\lambda & P_{a,b|x,y,\lambda} \ & \varrho^{C}		 (\lambda)  \label{eq:biseparableassemblage3}
\end{eqnarray}\label{eq:biseparableassemblage}\end{subequations}
where the last sum can be any TO-LHS assemblage.
\end{dfn}

The difference from D. Cavalcanti \emph{et al}'s definition is that they consider assemblages obtained from a quantum realization with bi-separable states. Reproducing Eqs.\ (4,5,6) of \cite{Cavalcanti2015a}, a tripartite state $\varrho^{ABC}$ is bi-separable when decomposable as
\begin{subequations}\begin{eqnarray}
\varrho^{ABC}	= \sum_\mu p^{A|BC}_\mu &\ \varrho_\mu^{A}  &\otimes \varrho_\mu^{BC} \label{eq:biseparablestate1}\\
							+ \sum_\nu p^{B|AC}_\nu &\ \varrho_\nu^{B}  &\otimes \varrho_\nu^{AC}  \label{eq:biseparablestate2}\\
							+ \sum_\lambda p^{AB|C}_\lambda &\ \varrho_\lambda^{AB}  &\otimes \varrho_\lambda^{C} \ .\label{eq:biseparablestate3}
\end{eqnarray}\label{eq:biseparablestate}\end{subequations}
Under local measurements on the $A$ and $B$ partitions, this yields a 2DI+1DD assemblage of the form \eqref{eq:biseparableassemblage} (akin to Eqs.\ (7,8,9) of \cite{Cavalcanti2015a}), but
with a distribution $P_{a,b|x,y,\lambda}$ in Eq.\ \eqref{eq:biseparableassemblage3} necessarily quantum-realizable (a subset of NS distributions). In other words, they only allow the sum in Eq.\ \eqref{eq:biseparableassemblage3} to be quantum-realizable NS-LHS assemblages.
Our redefinition, then, reduces the set of genuinely multipartite steerable assemblages.

Morover, we show in Supplementary Notes \ref{app:strict_subset} that there are, in fact, quantum-realizable assemblages affected by this change. These assemblages are decomposable as in Eq.~\eqref{eq:biseparableassemblage} only with a TO-LHS (not NS-LHS) term in Eq.~\eqref{eq:biseparableassemblage3}, and hence their quantum realization requires genuinely multipartite entangled states [i.e.\ not decomposable as Eq.~\eqref{eq:biseparablestate}]. 
Finally, once again entanglement certification is dissociated from steering: the TO-LHS assemblage from Supplementary Notes \ref{app:strict_subset} can be written in a bi-separable decomposition as in Eq.~\eqref{eq:biseparableassemblage}, but cannot be obtained from a bi-separable state as in Eq.~\eqref{eq:biseparablestate}. In other words, it certifies \emph{genuine} multipartite entanglement without steering.

\section{No-go theorem for multi-black-box universal steering bits}
\label{app:proof_nobits}
In contrast to the protocols exploring the capabilities of wirings within the $AB$ partition, in this section we present a no-go theorem limiting their transformation power.
Since it is known \cite{Gallego2015} that in minimal dimension there is no steering bit --- i.e. no ``universal'' minimal-dimension assemblage that can be transformed into any other under 1W-LOCCs --- one can ask whether reduction from a higher number of inputs, outputs or parties allows such a steering bit to be established. We answer in the negative even in minimal dimension. 

\begin{thm}[No pure steering bit with higher number of parties]
\label{th:nobits}
There does not exist any pure $(N-1)$-DI qubit assemblage $\sigma_{\boldsymbol a|\boldsymbol x}^{\text{bit}}$, where $\boldsymbol a=\{a_1,...,a_{N-1}\}$, $\boldsymbol x=\{x_1,...,x_{N-1}\}$ (with finite sets of input and output values), that can be transformed via 1W-LOCCs into all qubit assemblages of minimal dimension $\sigma_{a|x}^{(\text{target})}$.
\end{thm}

\begin{proof}
The proof is similar in spirit to that of Theorem 5 of \cite{Gallego2015}. We consider a pure $(N-1)$-DI qubit assemblage as a candidate for higher-dimensional ``bit'' assemblage. With the more detailed notation of \cite{Gallego2015}, it reads 
\begin{align}
\sigma_{\boldsymbol a|\boldsymbol x}^{\text{bit}}= P_{\boldsymbol A|\boldsymbol X}(\boldsymbol a|\boldsymbol x) \ \ket{\psi({\boldsymbol a,\boldsymbol x})}\bra{\psi({\boldsymbol a,\boldsymbol x})} \ . \label{eq:N-1DI} 
\end{align}
We assume the NS principle only between the DD party and all others, the $N-1$ DI parties may signal to each other at will. We will show that no single choice of $\sigma_{\boldsymbol a|\boldsymbol x}^{\text{bit}}$ can be freely transformed into members of a family of minimal-dimension assemblages $\sigma_{a_f|x_f}^\theta=\frac12\ket{\psi^\theta(a_f,x_f)}\bra{\psi^\theta(a_f,x_f)}$ for all $\theta\in{}]0,\pi/2[{}$, where
\begin{subequations}\begin{align}
\ket{\psi^\theta(0,0)} &= \ket{0} \label{eq:psitheta00}\\
\ket{\psi^\theta(1,0)} &= \ket{1} \label{eq:psitheta10}\\
\ket{\psi^\theta(0,1)} &=\ \ \cos\theta\ket{0} + \sin\theta\ket1 \label{eq:psitheta01}\\
\ket{\psi^\theta(1,1)} &= -\sin\theta\ket{0} + \cos\theta\ket1 \ . \label{eq:psitheta11}
\end{align}\label{eq:psitheta}\end{subequations}

The most general form of a 1W-LOCC applied to $\sigma_{\boldsymbol a|\boldsymbol x}^{\text{bit}}$ is 
\begin{widetext}
\begin{equation}
\sum_{\boldsymbol a, \boldsymbol x, \omega} P_{\boldsymbol X|X_f,\Omega}^\theta(\boldsymbol x|x_f,\omega) \ P_{A_f|\boldsymbol A,\boldsymbol X,\Omega,X_f}^\theta(a_f|\boldsymbol a,\boldsymbol x,\omega,x_f) \ P_{\boldsymbol A|\boldsymbol X}(\boldsymbol a|\boldsymbol x) \ K_\omega^\theta\ket{\psi({\boldsymbol a,\boldsymbol x})}\bra{\psi({\boldsymbol a,\boldsymbol x})}K_\omega^{\theta\dagger} \ ,
\label{eq:1WLOCCapplied}
\end{equation}
where $\Omega$ is a variable (with values $\omega$) representing information sent by the quantum party to the classical ones, $P_{\boldsymbol X|X_f,\Omega}^\theta$ and $P_{A_f|\boldsymbol A,\boldsymbol X,\Omega,X_f}^\theta$ are conditional probability distributions, and $K_\omega^\theta$ is a Kraus operator \cite{Gallego2015}; the three may depend on $\theta$. Since this transformed assemblage is intended to equal the rank-1 assemblage $\sigma_{a_f|x_f}^\theta$, we can conclude that  $\forall \ a_f,x_f$
\begin{multline}
\sum_{\boldsymbol a, \boldsymbol x} P_{\boldsymbol X|X_f,\Omega}^\theta(\boldsymbol x|x_f,\omega)P_{A_f|\boldsymbol A,\boldsymbol X,\Omega,X_f}^\theta(a_f|\boldsymbol a,\boldsymbol x,\omega,x_f) P_{\boldsymbol A|\boldsymbol X}(\boldsymbol a|\boldsymbol x) \ K_\omega^\theta\ket{\psi({\boldsymbol a,\boldsymbol x})}\bra{\psi({\boldsymbol a,\boldsymbol x})}K_\omega^{\theta\dagger} \\ \sim \ket{\psi^\theta(a_f,x_f)}\bra{\psi^\theta(a_f,x_f)} \ ,
\label{eq:1WLOCCequality}
\end{multline}\end{widetext}
where $\sim$ signifies ``is either null or proportional to'' and we have used the fact that the relation, valid for the sum in $\omega$, is also valid for each $\omega$ term. 

We will assume for now that $\sigma_{\boldsymbol a|\boldsymbol x}^{\text{bit}}$ is not a single-state assemblage, i.e., there is no state $\ket{\psi_{\text{single}}}$ such that $\ket{\psi(\boldsymbol a,\boldsymbol x)}=\ket{\psi_{\text{single}}}$ for all $\boldsymbol a,\boldsymbol x$ (for our purposes throughout this proof, states are equal if they differ only by an global phase). 

We now notice that, due to normalization, $\forall \ x_f,\omega$, $\exists \ \tilde{\boldsymbol x}, \tilde{\boldsymbol a}, \tilde a_f$ such that $P_{\boldsymbol X|X_f,\Omega}^\theta(\tilde{\boldsymbol x}|x_f,\omega)\times P_{A_f|\boldsymbol A,\boldsymbol X,\Omega,X_f}^\theta(\tilde a_f|\tilde{\boldsymbol a},\tilde{\boldsymbol x},\omega,x_f) \times P_{\boldsymbol A|\boldsymbol X}(\tilde{\boldsymbol a}|\tilde{\boldsymbol x})\neq0$. For these values, then,
\begin{equation}
K_\omega^\theta\ket{\psi(\tilde{\boldsymbol a},\tilde{\boldsymbol x})}\sim\ket{\psi^\theta(\tilde a_f,x_f)} \ .
\label{eq:Kpsi}
\end{equation}
In fact, there must be at least two different values $\tilde{\boldsymbol a}$ for each $\tilde{\boldsymbol x}$ for which Eq.\ \eqref{eq:Kpsi} is true, with the corresponding pure states $\ket{\psi(\tilde{\boldsymbol a},\tilde{\boldsymbol x})}$ being not all equal: if, for some $\tilde{\boldsymbol x}$ there is a single $\tilde{\boldsymbol a}$ with $P_{\boldsymbol A|\boldsymbol X}(\tilde{\boldsymbol a}|\tilde{\boldsymbol x})\neq0$, then by purity and the NS property between the DD and DI partitions, $\sigma_{\boldsymbol a|\boldsymbol x}^{\text{bit}}$ would be a single-state assemblage; if for all values $\tilde{\boldsymbol a}$, $\ket{\psi(\tilde{\boldsymbol a},\tilde{\boldsymbol x})}$ is the same, it would also be a single-state assemblage due to NS and purity.

Let us now exclude the possibility of $K_\omega^\theta\ket{\psi(\boldsymbol a,\boldsymbol x)}=0$ with $K_\omega^\theta\neq0$. If that were the case, $K_\omega^\theta$ would have a rank-1 support, hence a rank-1 span: $K_\omega^\theta\ket{\psi(\boldsymbol a,\boldsymbol x)}\sim \ket{k_\omega^\theta} \ \forall \ \boldsymbol a,\boldsymbol x$. From \eqref{eq:1WLOCCequality} and the independence of $x_f$ from $\omega$, this would require either $\ket{\psi^\theta(a_f,0)}\propto \ket{k_\omega^\theta}\propto \ket{\psi^\theta(\tilde a_f,1)}$ [contradiction with Eq.\ \eqref{eq:psitheta}] or that, for some value of $x_f$, for the corresponding $\tilde{\boldsymbol x}$, $K_\omega^\theta\ket{\psi(\tilde{\boldsymbol a},\tilde{\boldsymbol x})}=0$ for all $\tilde{\boldsymbol a}$ with $P_{\boldsymbol A|\boldsymbol X}(\tilde{\boldsymbol a}|\tilde{\boldsymbol x})\neq0$ [contradiction with there existing two different states $\ket{\psi(\tilde{\boldsymbol a},\tilde{\boldsymbol x})}$].

Finally, we can conclude from the dependencies of the three probabilities $P_{\boldsymbol X|X_f,\Omega}^\theta, P_{A_f|\boldsymbol A,\boldsymbol X,\Omega,X_f}^\theta, P_{\boldsymbol A|\boldsymbol X}$ on $x_f,\omega,\boldsymbol x,\boldsymbol a,\tilde a_f$, that 
\begin{equation}
K_\omega^\theta\ket{\psi(\tilde{\boldsymbol a},\tilde{\boldsymbol x})}\propto\ket{\psi^\theta(\tilde a_f,x_f)} \ .
\label{eq:Kpsiprop}
\end{equation}
The validity conditions of this equation are as follows: for all $(x_f,\omega)$, there exists some value $\tilde{\boldsymbol x}$ for which \eqref{eq:Kpsiprop} holds; for each $\tilde{\boldsymbol x}$, there are at least two values $\tilde{\boldsymbol a}$ for which \eqref{eq:Kpsiprop} holds; and for each choice of $(x_f,\omega,\tilde{\boldsymbol x},\tilde{\boldsymbol a})$ there is some value $\tilde a_f$ for which \eqref{eq:Kpsiprop} holds. Moreover, for given $\tilde{\boldsymbol x}$, the corresponding $\ket{\psi(\tilde{\boldsymbol a},\tilde{\boldsymbol x})}$ (for varying $\tilde{\boldsymbol a}$) are not all equal.

Let us explore the possible ways of satisfying Eq.\ \eqref{eq:Kpsiprop} by case analysis. A first possibility is that, for the two different values $x_f=0,1$, the values of $\tilde{\boldsymbol x}$ for which \eqref{eq:Kpsiprop} holds intersect at some value $\tilde{\boldsymbol x}_{\text{int}}$. Then $\exists \ \tilde{\boldsymbol a}, \tilde a_{f0}, \tilde a_{f1}$ such that
\begin{equation}\begin{split}
K_\omega^\theta\ket{\psi(\tilde{\boldsymbol a},\tilde{\boldsymbol x}_{\text{int}})}&\propto\ket{\psi^\theta(\tilde a_{f0},x_f=0)} \ , \\
K_\omega^\theta\ket{\psi(\tilde{\boldsymbol a},\tilde{\boldsymbol x}_{\text{int}})}&\propto\ket{\psi^\theta(\tilde a_{f1},x_f=1)} \ ,
\end{split}\label{eq:x_int}
\end{equation}
which is incompatible with Eq.\ \eqref{eq:psitheta}. We are then left with the values $\tilde{\boldsymbol x}$ for $x_f=0$ and $x_f=1$ being all different. Taking the liberty to relabel our variables, let us consider a value $\tilde{\boldsymbol x}=\boldsymbol 0$ for $x_f=0$ and a value $\tilde{\boldsymbol x}=\boldsymbol 1$ for $x_f=1$, ignoring the other possible values of $\tilde{\boldsymbol x}$ for which Eq.\ \eqref{eq:Kpsiprop} holds.
Let us call $\tilde{\boldsymbol a}=\boldsymbol0$ and $\tilde{\boldsymbol a}=\boldsymbol1$ the two values of $\tilde{\boldsymbol a}$ for which, given $\tilde{\boldsymbol x}$, Eq.\ \eqref{eq:Kpsiprop} holds. 
We see that $\tilde a_f$ could take any value for each $\tilde{\boldsymbol a}$. However, if $\tilde a_f$ is the same for the same $(x_f,\tilde{\boldsymbol x})$ and two different $\tilde{\boldsymbol a}$, e.g.,
\begin{equation}\begin{split}
&K_{\omega}\ket{\psi(\boldsymbol0,\boldsymbol1)} \propto \ket{\psi^\theta(0,1)} \\
&K_{\omega}\ket{\psi(\boldsymbol1,\boldsymbol1)} \propto \ket{\psi^\theta(0,1)} \ ,
\end{split}\label{eq:cases_b_c}\end{equation}
then Eq.\ \eqref{eq:Kpsiprop} cannot be satisfied for all $x_f$. This is because $\{\ket{\psi(\boldsymbol0,\boldsymbol1)},\ket{\psi(\boldsymbol1,\boldsymbol1)}\}$ form a basis of the qubit Hilbert space, hence $K_{\omega}$ has a 1-rank span given by $\ket{\psi^\theta(0,1)}$, which does not span $\ket{\psi^\theta(\tilde a_f,0)}$ as needed. Hence $\tilde a_f$ is different for each $\tilde{\boldsymbol a}$ value. 

We can then conclude that, up to relabeling, there must be states $\ket{\psi(\tilde{\boldsymbol a},\tilde{\boldsymbol x})}$ belonging to $\boldsymbol\sigma^{\text{bit}}$ which obey
\begin{subequations}\begin{align}
&K_{\omega}\ket{\psi(\boldsymbol0,\boldsymbol0)}		\propto		\ket{\psi^\theta(0,0)}	\label{eq:case_a00}\\
&K_{\omega}\ket{\psi(\boldsymbol1,\boldsymbol0)}		\propto		\ket{\psi^\theta(1,0)}	\label{eq:case_a10}\\
&K_{\omega}\ket{\psi(\boldsymbol0,\boldsymbol1)}		\propto		\ket{\psi^\theta(0,1)}	\label{eq:case_a01}\\
&K_{\omega}\ket{\psi(\boldsymbol1,\boldsymbol1)}		\propto		\ket{\psi^\theta(1,1)}	\label{eq:case_a11}
\end{align}\label{eq:case_a}\end{subequations}
to obtain the family of assemblages $\{\boldsymbol\sigma^{\theta}\}_{\theta\in{}]0,\pi/2[{}}$. We will choose the following parametrization:
\begin{equation}
\ket{\psi(\tilde{\boldsymbol a},\tilde{\boldsymbol x})} = \cos(\varphi_{\tilde{\boldsymbol a},\tilde{\boldsymbol x}})\ket0+e^{i\alpha_{\tilde{\boldsymbol a},\tilde{\boldsymbol x}}}\sin(\varphi_{\tilde{\boldsymbol a},\tilde{\boldsymbol x}})\ket1 \ , 
\label{eq:parametrization}
\end{equation}
where $\varphi_{\tilde{\boldsymbol a},\tilde{\boldsymbol x}}\in[0,\pi/2]$. It should be noted that $(\varphi_{\tilde{\boldsymbol a},\tilde{\boldsymbol x}},\alpha_{\tilde{\boldsymbol a},\tilde{\boldsymbol x}})$ may depend on $\theta$ through $\tilde{\boldsymbol a},\tilde{\boldsymbol x}$: because $P_{\boldsymbol X|X_f,\Omega}^\theta$ may depend on $\theta$, the values $\tilde{\boldsymbol a},\tilde{\boldsymbol x}$ for which Eq.\ \eqref{eq:Kpsiprop} holds may vary for different values of $\theta$. However, for finitely many values of $\boldsymbol a$, $\boldsymbol x$, there are only finitely many states and finitely many $(\varphi_{{\boldsymbol a},{\boldsymbol x}},\alpha_{{\boldsymbol a},{\boldsymbol x}})$ to pick from, so some choice of states as in Eq.\ \eqref{eq:parametrization} must still be able to satisfy Eq.\ \eqref{eq:case_a} for a continuous set of values $\theta$.

Substituting Eqs.\ \eqref{eq:psitheta} and \eqref{eq:parametrization} in (\ref{eq:case_a00},\ref{eq:case_a10}), respectively, we see that
\begin{equation}
\frac{K^\theta_{\omega00}}{K^\theta_{\omega01}} = - e^{i\alpha_{\boldsymbol1\boldsymbol0}}\tan\varphi_{\boldsymbol1\boldsymbol0}; \ \frac{K^\theta_{\omega10}}{K^\theta_{\omega11}} = - e^{i\alpha_{\boldsymbol0\boldsymbol0}}\tan\varphi_{\boldsymbol0\boldsymbol0};
\label{eq:Kraus_1}
\end{equation}
where $K^\theta_{\omega ij}:=\braket{i}{K^\theta_\omega|j}$. Doing the same in (\ref{eq:case_a01},\ref{eq:case_a11}) and substituting \eqref{eq:Kraus_1}, we find, respectively,
\begin{align}
\frac{K^\theta_{\omega11}}{K^\theta_{\omega01}} & = \tan\theta \frac{\tan\varphi_{\boldsymbol0\boldsymbol1}e^{i\alpha_{\boldsymbol0\boldsymbol1}}-\tan\varphi_{\boldsymbol1\boldsymbol0}e^{i\alpha_{\boldsymbol1\boldsymbol0}}}{\tan\varphi_{\boldsymbol0\boldsymbol1}e^{i\alpha_{\boldsymbol0\boldsymbol1}}+\tan\varphi_{\boldsymbol0\boldsymbol0}e^{i\alpha_{\boldsymbol0\boldsymbol0}}}
\label{eq:Kraus_2} \\
\frac{K^\theta_{\omega11}}{K^\theta_{\omega01}} & = \frac{-1}{\tan\theta} \frac{\tan\varphi_{\boldsymbol1\boldsymbol1}e^{i\alpha_{\boldsymbol1\boldsymbol1}}-\tan\varphi_{\boldsymbol1\boldsymbol0}e^{i\alpha_{\boldsymbol1\boldsymbol0}}}{\tan\varphi_{\boldsymbol1\boldsymbol1}e^{i\alpha_{\boldsymbol1\boldsymbol1}}-\tan\varphi_{\boldsymbol0\boldsymbol0}e^{i\alpha_{\boldsymbol0\boldsymbol0}}} .
\label{eq:Kraus_3}
\end{align}
Equating the two, we have
\begin{equation}\begin{split}
\tan^2\theta\left(\frac{\tan\varphi_{\boldsymbol0\boldsymbol1}e^{i\alpha_{\boldsymbol0\boldsymbol1}}-\tan\varphi_{\boldsymbol1\boldsymbol0}e^{i\alpha_{\boldsymbol1\boldsymbol0}}}{\tan\varphi_{\boldsymbol0\boldsymbol1}e^{i\alpha_{\boldsymbol0\boldsymbol1}}+\tan\varphi_{\boldsymbol0\boldsymbol0}e^{i\alpha_{\boldsymbol0\boldsymbol0}}}\right)+\\
+\left(\frac{\tan\varphi_{\boldsymbol1\boldsymbol1}e^{i\alpha_{\boldsymbol1\boldsymbol1}}-\tan\varphi_{\boldsymbol1\boldsymbol0}e^{i\alpha_{\boldsymbol1\boldsymbol0}}}{\tan\varphi_{\boldsymbol1\boldsymbol1}e^{i\alpha_{\boldsymbol1\boldsymbol1}}-\tan\varphi_{\boldsymbol0\boldsymbol0}e^{i\alpha_{\boldsymbol0\boldsymbol0}}}\right)=0 \ ,
\end{split}\label{eq:tantheta}\end{equation}
which, for fixed $\varphi_{\tilde{\boldsymbol a},\tilde{\boldsymbol x}},\alpha_{\tilde{\boldsymbol a},\tilde{\boldsymbol x}}$, must hold for a continuous set of values $\theta$. This is only possible if both parentheses are zero, which in turn implies $(\varphi_{\boldsymbol0,\boldsymbol1},\alpha_{\boldsymbol0,\boldsymbol1})=(\varphi_{\boldsymbol1,\boldsymbol0},\alpha_{\boldsymbol1,\boldsymbol0})=(\varphi_{\boldsymbol1,\boldsymbol1},\alpha_{\boldsymbol1,\boldsymbol1})$, or $\ket{\psi(\boldsymbol0,\boldsymbol1)}=\ket{\psi(\boldsymbol1,\boldsymbol0)}=\ket{\psi(\boldsymbol1,\boldsymbol1)}$, contradicting the established relation $\ket{\psi(\boldsymbol0,\boldsymbol1)}\neq\ket{\psi(\boldsymbol1,\boldsymbol1)}$. This concludes the demonstration for non-single-state assemblages.

Finally, let us show that a single-state assemblage is unable to do the task. From \eqref{eq:1WLOCCapplied},
\begin{widetext}
\begin{equation}\begin{split}
\sum_{\boldsymbol a, \boldsymbol x} P_{\boldsymbol X|X_f,\Omega}^\theta(\boldsymbol x|x_f,\omega)&P_{A_f|\boldsymbol A,\boldsymbol X,\Omega,X_f}^\theta(a_f|\boldsymbol a,\boldsymbol x,\omega,x_f) P_{\boldsymbol A|\boldsymbol X}(\boldsymbol a|\boldsymbol x) \times \\
& \times K_\omega^\theta\ket{\psi_{\text{single}}}\bra{\psi_{\text{single}}}K_\omega^{\theta\dagger} \sim \ket{\psi^\theta(a_f,x_f)}\bra{\psi^\theta(a_f,x_f)} \ .
\end{split}\label{eq:1WLOCC_singlestate}
\end{equation}\end{widetext}
The sum on the left-hand side is not zero for at least two pairs $(a_f,x_f)$, hence $K_\omega^\theta\ket{\psi_{\text{single}}}$ must be proportional to $\ket{\psi^\theta(a_f,x_f)}$ for both these pairs. This is incompatible with Eq.\ \eqref{eq:psitheta}, since none of the $\ket{\psi^\theta(a_f,x_f)}$ are proportional to one another.
\end{proof}
}


\begin{thebibliography}{59}%
\makeatletter
\providecommand \@ifxundefined [1]{%
 \@ifx{#1\undefined}
}%
\providecommand \@ifnum [1]{%
 \ifnum #1\expandafter \@firstoftwo
 \else \expandafter \@secondoftwo
 \fi
}%
\providecommand \@ifx [1]{%
 \ifx #1\expandafter \@firstoftwo
 \else \expandafter \@secondoftwo
 \fi
}%
\providecommand \natexlab [1]{#1}%
\providecommand \enquote  [1]{``#1''}%
\providecommand \bibnamefont  [1]{#1}%
\providecommand \bibfnamefont [1]{#1}%
\providecommand \citenamefont [1]{#1}%
\providecommand \href@noop [0]{\@secondoftwo}%
\providecommand \href [0]{\begingroup \@sanitize@url \@href}%
\providecommand \@href[1]{\@@startlink{#1}\@@href}%
\providecommand \@@href[1]{\endgroup#1\@@endlink}%
\providecommand \@sanitize@url [0]{\catcode `\\12\catcode `\$12\catcode
  `\&12\catcode `\#12\catcode `\^12\catcode `\_12\catcode `\%12\relax}%
\providecommand \@@startlink[1]{}%
\providecommand \@@endlink[0]{}%
\providecommand \url  [0]{\begingroup\@sanitize@url \@url }%
\providecommand \@url [1]{\endgroup\@href {#1}{\urlprefix }}%
\providecommand \urlprefix  [0]{URL }%
\providecommand \Eprint [0]{\href }%
\providecommand \doibase [0]{http://dx.doi.org/}%
\providecommand \selectlanguage [0]{\@gobble}%
\providecommand \bibinfo  [0]{\@secondoftwo}%
\providecommand \bibfield  [0]{\@secondoftwo}%
\providecommand \translation [1]{[#1]}%
\providecommand \BibitemOpen [0]{}%
\providecommand \bibitemStop [0]{}%
\providecommand \bibitemNoStop [0]{.\EOS\space}%
\providecommand \EOS [0]{\spacefactor3000\relax}%
\providecommand \BibitemShut  [1]{\csname bibitem#1\endcsname}%
\let\auto@bib@innerbib\@empty
\bibitem [{\citenamefont {Horodecki}\ \emph {et~al.}(2009)\citenamefont
  {Horodecki}, \citenamefont {Horodecki}, \citenamefont {Horodecki},\ and\
  \citenamefont {Horodecki}}]{Horodecki2009}%
  \BibitemOpen
  \bibfield  {author} {\bibinfo {author} {\bibfnamefont {Ryszard}\ \bibnamefont
  {Horodecki}}, \bibinfo {author} {\bibfnamefont {Pawe{\l}}\ \bibnamefont
  {Horodecki}}, \bibinfo {author} {\bibfnamefont {Micha{\l}}\ \bibnamefont
  {Horodecki}}, \ and\ \bibinfo {author} {\bibfnamefont {Karol}\ \bibnamefont
  {Horodecki}},\ }\bibfield  {title} {\enquote {\bibinfo {title} {{Quantum
  entanglement}},}\ }\href {\doibase 10.1103/RevModPhys.81.865} {\bibfield
  {journal} {\bibinfo  {journal} {Rev. Mod. Phys.}\ }\textbf
  {\bibinfo {volume} {81}},\ \bibinfo {pages} {865--942} (\bibinfo {year}
  {2009})},\ \Eprint {http://arxiv.org/abs/quant-ph/0702225} {arXiv:quant-ph/0702225}
  \BibitemShut {NoStop}%
\bibitem [{\citenamefont {Brunner}\ \emph {et~al.}(2014)\citenamefont
  {Brunner}, \citenamefont {Cavalcanti}, \citenamefont {Pironio}, \citenamefont
  {Scarani},\ and\ \citenamefont {Wehner}}]{Brunner2014}%
  \BibitemOpen
  \bibfield  {author} {\bibinfo {author} {\bibfnamefont {Nicolas}\ \bibnamefont
  {Brunner}}, \bibinfo {author} {\bibfnamefont {Daniel}\ \bibnamefont
  {Cavalcanti}}, \bibinfo {author} {\bibfnamefont {Stefano}\ \bibnamefont
  {Pironio}}, \bibinfo {author} {\bibfnamefont {Valerio}\ \bibnamefont
  {Scarani}}, and\ \bibinfo {author} {\bibfnamefont {Stephanie}\ \bibnamefont
  {Wehner}},\ }\bibfield  {title} {\enquote {\bibinfo {title} {{Bell
  nonlocality}},}\ }\href {\doibase 10.1103/RevModPhys.86.419} {\bibfield
  {journal} {\bibinfo  {journal} {Rev. Mod. Phys.}\ }\textbf
  {\bibinfo {volume} {86}},\ \bibinfo {pages} {419--478} (\bibinfo {year}
  {2014})},\ \Eprint {http://arxiv.org/abs/1303.2849} {arXiv:1303.2849}
  \BibitemShut {NoStop}%
\bibitem [{\citenamefont {Reid}\ \emph {et~al.}(2009)\citenamefont {Reid},
  \citenamefont {Drummond}, \citenamefont {Bowen}, \citenamefont {Cavalcanti},
  \citenamefont {Lam}, \citenamefont {Bachor}, \citenamefont {Andersen}, and\
  \citenamefont {Leuchs}}]{Reid2009}%
  \BibitemOpen
  \bibfield  {author} {\bibinfo {author} {\bibfnamefont {M.~D.}\ \bibnamefont
  {Reid}} et al, }\bibfield  {title} {\enquote
  {\bibinfo {title} {{Colloquium : The Einstein-Podolsky-Rosen paradox: From
  concepts to applications}},}\ }\href {\doibase 10.1103/RevModPhys.81.1727}
  {\bibfield  {journal} {\bibinfo  {journal} {Rev. Mod. Phys.}\
  }\textbf {\bibinfo {volume} {81}},\ \bibinfo {pages} {1727--1751} (\bibinfo
  {year} {2009})},\ \Eprint {http://arxiv.org/abs/0806.0270} {arXiv:0806.0270}
  \BibitemShut {NoStop}%
\bibitem [{\citenamefont {Cavalcanti}\ and\ \citenamefont
  {Skrzypczyk}(2017)}]{Cavalcanti2017}%
  \BibitemOpen
  \bibfield  {author} {\bibinfo {author} {\bibfnamefont {D.}~\bibnamefont
  {Cavalcanti}}\ and\ \bibinfo {author} {\bibfnamefont {P.}~\bibnamefont
  {Skrzypczyk}},\ }\bibfield  {title} {\enquote {\bibinfo {title} {{Quantum
  steering: a review with focus on semidefinite programming}},}\ }\href
  {\doibase 10.1088/1361-6633/80/2/024001} {\bibfield  {journal} {\bibinfo
  {journal} {Rep. Prog. Phys.}\ }\textbf {\bibinfo {volume}
  {80}},\ \bibinfo {pages} {024001} (\bibinfo {year} {2017})},\ \Eprint
  {http://arxiv.org/abs/1604.00501} {arXiv:1604.00501} \BibitemShut {NoStop}%
\bibitem [{\citenamefont {Uola}\ \emph {et~al.}(2019)\citenamefont {Uola},
  \citenamefont {Costa}, \citenamefont {Nguyen},\ and\ \citenamefont
  {G{\"{u}}hne}}]{Uola2019}%
  \BibitemOpen
  \bibfield  {author} {\bibinfo {author} {\bibfnamefont {Roope}\ \bibnamefont
  {Uola}}, \bibinfo {author} {\bibfnamefont {Ana C.~S.}\ \bibnamefont {Costa}},
  \bibinfo {author} {\bibfnamefont {H.~Chau}\ \bibnamefont {Nguyen}}, and\
  \bibinfo {author} {\bibfnamefont {Otfried}\ \bibnamefont {G{\"{u}}hne}},
  } {\enquote {\bibinfo {title} {{Quantum Steering}},}\ }\href {\doibase 10.1103/RevModPhys.92.015001}
	{\bibfield  {journal} {\bibinfo  {journal} {Rev. Mod. Phys.}\
  }\textbf {\bibinfo {volume} {92}},\ \bibinfo {pages} {015001} (\bibinfo
  {year} {2020})}, 	\ \Eprint{http://arxiv.org/abs/1903.06663} {arXiv:1903.06663} \BibitemShut {NoStop}%
\bibitem [{\citenamefont {Barrett}\ \emph {et~al.}(2005)\citenamefont
  {Barrett}, \citenamefont {Hardy}, and\ \citenamefont {Kent}}]{Barrett2005}%
  \BibitemOpen
  \bibfield  {author} {\bibinfo {author} {\bibfnamefont {Jonathan}\
  \bibnamefont {Barrett}}, \bibinfo {author} {\bibfnamefont {Lucien}\
  \bibnamefont {Hardy}}, and\ \bibinfo {author} {\bibfnamefont {Adrian}\
  \bibnamefont {Kent}},\ }\bibfield  {title} {\enquote {\bibinfo {title} {{No
  Signaling and Quantum Key Distribution}},}\ }\href {\doibase 10.1103/PhysRevLett.95.010503} {\bibfield  {journal} {\bibinfo  {journal}
  {Phys. Rev. Lett.}\ }\textbf {\bibinfo {volume} {95}},\ \bibinfo
  {pages} {010503} (\bibinfo {year} {2005})},\ \Eprint
  {http://arxiv.org/abs/quant-ph/0405101} {arXiv:quant-ph/0405101} \BibitemShut
  {NoStop}%
\bibitem [{\citenamefont {Ac{\'{i}}n}\ \emph
  {et~al.}(2006{\natexlab{a}})\citenamefont {Ac{\'{i}}n}, \citenamefont
  {Gisin}, and\ \citenamefont {Masanes}}]{Acin2006}%
  \BibitemOpen
  \bibfield  {author} {\bibinfo {author} {\bibfnamefont {Antonio}\ \bibnamefont
  {Ac{\'{i}}n}}, \bibinfo {author} {\bibfnamefont {Nicolas}\ \bibnamefont
  {Gisin}}, and\ \bibinfo {author} {\bibfnamefont {Lluis}\ \bibnamefont
  {Masanes}},\ }\bibfield  {title} {\enquote {\bibinfo {title} {{From Bell's
  Theorem to Secure Quantum Key Distribution}},}\ }\href {\doibase
  10.1103/PhysRevLett.97.120405} {\bibfield  {journal} {\bibinfo  {journal}
  {Phys. Rev. Lett.}\ }\textbf {\bibinfo {volume} {97}},\ \bibinfo
  {pages} {120405} (\bibinfo {year} {2006}{\natexlab{a}})},\ \Eprint
  {http://arxiv.org/abs/quant-ph/0510094} {arXiv:quant-ph/0510094} \BibitemShut
  {NoStop}%
\bibitem [{\citenamefont {Ac{\'{i}}n}\ \emph
  {et~al.}(2006{\natexlab{b}})\citenamefont {Ac{\'{i}}n}, \citenamefont
  {Massar}, and\ \citenamefont {Pironio}}]{Acin2006a}%
  \BibitemOpen
  \bibfield  {author} {\bibinfo {author} {\bibfnamefont {Antonio}\ \bibnamefont
  {Ac{\'{i}}n}}, \bibinfo {author} {\bibfnamefont {Serge}\ \bibnamefont
  {Massar}}, and\ \bibinfo {author} {\bibfnamefont {Stefano}\ \bibnamefont
  {Pironio}},\ }\bibfield  {title} {\enquote {\bibinfo {title} {{Efficient
  quantum key distribution secure against no-signalling eavesdroppers}},}\
  }\href {\doibase 10.1088/1367-2630/8/8/126} {\bibfield  {journal} {\bibinfo
  {journal} {New J. Phys.}\ }\textbf {\bibinfo {volume} {8}},\
  \bibinfo {pages} {126} (\bibinfo {year} {2006}{\natexlab{b}})},\ \Eprint
  {http://arxiv.org/abs/quant-ph/0605246} {arXiv:quant-ph/0605246} \BibitemShut
  {NoStop}%
\bibitem [{\citenamefont {Ac{\'{i}}n}\ \emph {et~al.}(2007)\citenamefont
  {Ac{\'{i}}n}, \citenamefont {Brunner}, \citenamefont {Gisin}, \citenamefont
  {Massar}, \citenamefont {Pironio}, and\ \citenamefont {Scarani}}]{Acin2007}%
  \BibitemOpen
  \bibfield  {author} {\bibinfo {author} {\bibfnamefont {Antonio}\ \bibnamefont
  {Ac{\'{i}}n}} et al, } \bibfield  {title} {\enquote {\bibinfo {title} {{Device-Independent Security
  of Quantum Cryptography against Collective Attacks}},}\ }\href {\doibase
  10.1103/PhysRevLett.98.230501} {\bibfield  {journal} {\bibinfo  {journal}
  {Phys. Rev. Lett.}\ }\textbf {\bibinfo {volume} {98}},\ \bibinfo
  {pages} {230501} (\bibinfo {year} {2007})},\ \Eprint
  {http://arxiv.org/abs/quant-ph/0702152} {arXiv:quant-ph/0702152 } \BibitemShut
  {NoStop}%
\bibitem [{\citenamefont {Colbeck}(2006)}]{Colbeck2009}%
  \BibitemOpen
  \bibfield  {author} {\bibinfo {author} {\bibfnamefont {Roger}\ \bibnamefont
  {Colbeck}},\ }\emph {\bibinfo {title} {{Quantum And Relativistic Protocols
  For Secure Multi-Party Computation}}},\ \href
  {http://arxiv.org/abs/0911.3814} {Ph.D. thesis},\ \bibinfo  {school}
  {University of Cambridge} (\bibinfo {year} {2006}),\ \Eprint
  {http://arxiv.org/abs/0911.3814} {arXiv:0911.3814} \BibitemShut {NoStop}%
\bibitem [{\citenamefont {Colbeck} and\ \citenamefont
  {Kent}(2011)}]{Colbeck2010}%
  \BibitemOpen
  \bibfield  {author} {\bibinfo {author} {\bibfnamefont {Roger}\ \bibnamefont
  {Colbeck}} and\ \bibinfo {author} {\bibfnamefont {Adrian}\ \bibnamefont
  {Kent}},\ }\bibfield  {title} {\enquote {\bibinfo {title} {{Private
  randomness expansion with untrusted devices}},}\ }\href {\doibase 10.1088/1751-8113/44/9/095305} {\bibfield  {journal} {\bibinfo  {journal}
  {J. Phys. A Math. Theor.}\ }\textbf {\bibinfo
  {volume} {44}},\ \bibinfo {pages} {095305} (\bibinfo {year} {2011})},\
  \Eprint {http://arxiv.org/abs/1011.4474} {arXiv:1011.4474} \BibitemShut
  {NoStop}%
\bibitem [{\citenamefont {Pironio}\ \emph {et~al.}(2010)\citenamefont
  {Pironio}, \citenamefont {Ac{\'{i}}n}, \citenamefont {Massar}, \citenamefont
  {de~la Giroday}, \citenamefont {Matsukevich}, \citenamefont {Maunz},
  \citenamefont {Olmschenk}, \citenamefont {Hayes}, \citenamefont {Luo},
  \citenamefont {Manning}, and\ \citenamefont {Monroe}}]{Pironio2010}%
  \BibitemOpen
  \bibfield  {author} {\bibinfo {author} {\bibfnamefont {S.}~\bibnamefont
  {Pironio}} et al,} \bibfield  {title} {\enquote {\bibinfo {title}
  {{Random numbers certified by Bell's theorem}},}\ }\href {\doibase
  10.1038/nature09008} {\bibfield  {journal} {\bibinfo  {journal} {Nature}\
  }\textbf {\bibinfo {volume} {464}},\ \bibinfo {pages} {1021--1024} (\bibinfo
  {year} {2010})},\ \Eprint {http://arxiv.org/abs/0911.3427} {arXiv:0911.3427}
  \BibitemShut {NoStop}%
\bibitem [{\citenamefont {Ac{\'{i}}n} and\ \citenamefont
  {Masanes}(2016)}]{Acin2016}%
  \BibitemOpen
  \bibfield  {author} {\bibinfo {author} {\bibfnamefont {Antonio}\ \bibnamefont
  {Ac{\'{i}}n}}\ and\ \bibinfo {author} {\bibfnamefont {Lluis}\ \bibnamefont
  {Masanes}},\ }\bibfield  {title} {\enquote {\bibinfo {title} {{Certified
  randomness in quantum physics}},}\ }\href {\doibase 10.1038/nature20119}
  {\bibfield  {journal} {\bibinfo  {journal} {Nature}\ }\textbf {\bibinfo
  {volume} {540}},\ \bibinfo {pages} {213--219} (\bibinfo {year} {2016})},\
  \Eprint {http://arxiv.org/abs/1708.00265} {arXiv:1708.00265} \BibitemShut
  {NoStop}%
\bibitem [{\citenamefont {Gheorghiu}\ \emph {et~al.}(2015)\citenamefont
  {Gheorghiu}, \citenamefont {Kashefi},\ and\ \citenamefont
  {Wallden}}]{Gheorghiu2015}%
  \BibitemOpen
  \bibfield  {author} {\bibinfo {author} {\bibfnamefont {Alexandru}\
  \bibnamefont {Gheorghiu}}, \bibinfo {author} {\bibfnamefont {Elham}\
  \bibnamefont {Kashefi}}, and\ \bibinfo {author} {\bibfnamefont {Petros}\
  \bibnamefont {Wallden}},\ }\bibfield  {title} {\enquote {\bibinfo {title}
  {{Robustness and device independence of verifiable blind quantum
  computing}},}\ }\href {\doibase 10.1088/1367-2630/17/8/083040} {\bibfield
  {journal} {\bibinfo  {journal} {New J. Phys.}\ }\textbf {\bibinfo
  {volume} {17}},\ \bibinfo {pages} {083040} (\bibinfo {year} {2015})},\
  \Eprint {http://arxiv.org/abs/1502.02571} {arXiv:1502.02571} \BibitemShut
  {NoStop}%
\bibitem [{\citenamefont {Hajdu{\v{s}}ek}\ \emph {et~al.}(2015)\citenamefont
  {Hajdu{\v{s}}ek}, \citenamefont {P{\'{e}}rez-Delgado},\ and\ \citenamefont
  {Fitzsimons}}]{Hajdusek2015}%
  \BibitemOpen
  \bibfield  {author} {\bibinfo {author} {\bibfnamefont {Michal}\ \bibnamefont
  {Hajdu{\v{s}}ek}}, \bibinfo {author} {\bibfnamefont {Carlos~A.}\ \bibnamefont
  {P{\'{e}}rez-Delgado}}, and\ \bibinfo {author} {\bibfnamefont {Joseph~F.}\
  \bibnamefont {Fitzsimons}},\ }\href {http://arxiv.org/abs/1502.02563}
  {\enquote {\bibinfo {title} {{Device-Independent Verifiable Blind Quantum
  Computation}}.} } Preprint at  \Eprint
  {http://arxiv.org/abs/1502.02563} {http://arxiv.org/abs/1502.02563} (\bibinfo {year} {2015})\BibitemShut {NoStop}%
\bibitem [{\citenamefont {Ribeiro}\ \emph {et~al.}(2018)\citenamefont
  {Ribeiro}, \citenamefont {Murta},\ and\ \citenamefont
  {Wehner}}]{Ribeiro2018}%
  \BibitemOpen
  \bibfield  {author} {\bibinfo {author} {\bibfnamefont {J{\'{e}}r{\'{e}}my}\
  \bibnamefont {Ribeiro}}, \bibinfo {author} {\bibfnamefont {Gl{\'{a}}ucia}\
  \bibnamefont {Murta}}, and\ \bibinfo {author} {\bibfnamefont {Stephanie}\
  \bibnamefont {Wehner}},\ }\bibfield  {title} {\enquote {\bibinfo {title}
  {{Fully device-independent conference key agreement}},}\ }\href {\doibase 10.1103/PhysRevA.97.022307} {\bibfield  {journal} {\bibinfo  {journal}
  {Phys. Rev. A}\ }\textbf {\bibinfo {volume} {97}},\ \bibinfo {pages}
  {022307} (\bibinfo {year} {2018})},\ \Eprint
  {http://arxiv.org/abs/1708.00798} {arXiv:1708.00798} \BibitemShut {NoStop}%
\bibitem [{\citenamefont {Holz}\ \emph {et~al.}(2020)\citenamefont {Holz},
  \citenamefont {Kampermann},\ and\ \citenamefont {Bru{\ss}}}]{Holz2020}%
  \BibitemOpen
  \bibfield  {author} {\bibinfo {author} {\bibfnamefont {Timo}\ \bibnamefont
  {Holz}}, \bibinfo {author} {\bibfnamefont {Hermann}\ \bibnamefont
  {Kampermann}}, and\ \bibinfo {author} {\bibfnamefont {Dagmar}\ \bibnamefont
  {Bru{\ss}}},\ }\bibfield  {title} {\enquote {\bibinfo {title} {{Genuine
  multipartite Bell inequality for device-independent conference key
  agreement}},}\ }\href {\doibase 10.1103/PhysRevResearch.2.023251} {\bibfield
  {journal} {\bibinfo  {journal} {Phys. Rev. Res.}\ }\textbf {\bibinfo
  {volume} {2}},\ \bibinfo {pages} {023251} (\bibinfo {year} {2020})},\ \Eprint
  {http://arxiv.org/abs/1910.11360} {arXiv:1910.11360} \BibitemShut {NoStop}%
\bibitem [{\citenamefont {Murta}\ \emph {et~al.}(2020)\citenamefont {Murta},
  \citenamefont {Grasselli}, \citenamefont {Kampermann},\ and\ \citenamefont
  {Bru{\ss}}}]{Murta2020}%
  \BibitemOpen
  \bibfield  {author} {\bibinfo {author} {\bibfnamefont {Gl{\'{a}}ucia}\
  \bibnamefont {Murta}}, \bibinfo {author} {\bibfnamefont {Federico}\
  \bibnamefont {Grasselli}}, \bibinfo {author} {\bibfnamefont {Hermann}\
  \bibnamefont {Kampermann}}, and\ \bibinfo {author} {\bibfnamefont {Dagmar}\
  \bibnamefont {Bru{\ss}}},\ }\bibfield  {title} {\enquote {\bibinfo {title}
  {{Quantum Conference Key Agreement: A Review}},}\ }\href {\doibase 10.1002/qute.202000025} {\bibfield  {journal} {\bibinfo  {journal} {Adv. Quantum Technol.}\ }\textbf {\bibinfo {volume} {3}},\ \bibinfo {pages}
  {2000025} (\bibinfo {year} {2020})},\ \Eprint
  {http://arxiv.org/abs/2003.10186} {arXiv:2003.10186} \BibitemShut {NoStop}%
\bibitem [{\citenamefont {Wiseman}\ \emph {et~al.}(2006)\citenamefont
  {Wiseman}, \citenamefont {Jones},\ and\ \citenamefont
  {Doherty}}]{Wiseman2007}%
  \BibitemOpen
  \bibfield  {author} {\bibinfo {author} {\bibfnamefont {H.~M.}\ \bibnamefont
  {Wiseman}}, \bibinfo {author} {\bibfnamefont {S.~J.}\ \bibnamefont {Jones}}, \
  and\ \bibinfo {author} {\bibfnamefont {A.~C.}\ \bibnamefont {Doherty}},\
  }\bibfield  {title} {\enquote {\bibinfo {title} {{Steering, Entanglement,
  Nonlocality, and the EPR Paradox}},}\ }\href {\doibase 10.1103/PhysRevLett.98.140402} {\bibfield  {journal} {\bibinfo  {journal}
  {Phys. Rev. Lett.}\ }\textbf {\bibinfo {volume} {98}},\ \bibinfo
  {pages} {140402} (\bibinfo {year} {2006})},\ \Eprint
  {http://arxiv.org/abs/quant-ph/0612147} {arXiv:quant-ph/0612147 } \BibitemShut
  {NoStop}%
\bibitem [{\citenamefont {Jones}\ \emph {et~al.}(2007)\citenamefont {Jones},
  \citenamefont {Wiseman},\ and\ \citenamefont {Doherty}}]{Jones2007}%
  \BibitemOpen
  \bibfield  {author} {\bibinfo {author} {\bibfnamefont {S.~J.}\ \bibnamefont
  {Jones}}, \bibinfo {author} {\bibfnamefont {H.~M.}\ \bibnamefont {Wiseman}},
  \ and\ \bibinfo {author} {\bibfnamefont {A.~C.}\ \bibnamefont {Doherty}},\
  }\bibfield  {title} {\enquote {\bibinfo {title} {{Entanglement,
  Einstein-Podolsky-Rosen correlations, Bell nonlocality, and steering}},}\
  }\href {\doibase 10.1103/PhysRevA.76.052116} {\bibfield  {journal} {\bibinfo
  {journal} {Phys. Rev. A}\ }\textbf {\bibinfo {volume} {76}},\ \bibinfo
  {pages} {052116} (\bibinfo {year} {2007})},\ \Eprint
  {http://arxiv.org/abs/0709.0390v2} {arXiv:0709.0390v2} \BibitemShut {NoStop}%
\bibitem [{\citenamefont {Branciard}\ \emph {et~al.}(2012)\citenamefont
  {Branciard}, \citenamefont {Cavalcanti}, \citenamefont {Walborn},
  \citenamefont {Scarani},\ and\ \citenamefont {Wiseman}}]{Branciard2012}%
  \BibitemOpen
  \bibfield  {author} {\bibinfo {author} {\bibfnamefont {Cyril}\ \bibnamefont
  {Branciard}}, \bibinfo {author} {\bibfnamefont {Eric~G.}\ \bibnamefont
  {Cavalcanti}}, \bibinfo {author} {\bibfnamefont {Stephen~P.}\ \bibnamefont
  {Walborn}}, \bibinfo {author} {\bibfnamefont {Valerio}\ \bibnamefont
  {Scarani}}, and\ \bibinfo {author} {\bibfnamefont {Howard~M.}\ \bibnamefont
  {Wiseman}},\ }\bibfield  {title} {\enquote {\bibinfo {title} {{One-sided
  device-independent quantum key distribution: Security, feasibility, and the
  connection with steering}},}\ }\href {\doibase 10.1103/PhysRevA.85.010301}
  {\bibfield  {journal} {\bibinfo  {journal} {Phys. Rev. A}\ }\textbf
  {\bibinfo {volume} {85}},\ \bibinfo {pages} {010301} (\bibinfo {year}
  {2012})},\ \Eprint {http://arxiv.org/abs/1109.1435} {arXiv:1109.1435}
  \BibitemShut {NoStop}%
\bibitem [{\citenamefont {He}\ and\ \citenamefont {Reid}(2013)}]{He2013}%
  \BibitemOpen
  \bibfield  {author} {\bibinfo {author} {\bibfnamefont {Q.~Y.}\ \bibnamefont
  {He}}\ and\ \bibinfo {author} {\bibfnamefont {M.~D.}\ \bibnamefont {Reid}},\
  }\bibfield  {title} {\enquote {\bibinfo {title} {{Genuine Multipartite
  Einstein-Podolsky-Rosen Steering}},}\ }\href {\doibase 10.1103/PhysRevLett.111.250403} {\bibfield  {journal} {\bibinfo  {journal}
  {Phys. Rev. Lett.}\ }\textbf {\bibinfo {volume} {111}},\ \bibinfo
  {pages} {250403} (\bibinfo {year} {2013})},\ \Eprint
  {http://arxiv.org/abs/1212.2270} {arXiv:1212.2270} \BibitemShut {NoStop}%
\bibitem [{\citenamefont {Skrzypczyk}\ and\ \citenamefont
  {Cavalcanti}(2018)}]{Skrzypczyk2018}%
  \BibitemOpen
  \bibfield  {author} {\bibinfo {author} {\bibfnamefont {Paul}\ \bibnamefont
  {Skrzypczyk}}\ and\ \bibinfo {author} {\bibfnamefont {Daniel}\ \bibnamefont
  {Cavalcanti}},\ }\bibfield  {title} {\enquote {\bibinfo {title} {{Maximal
  Randomness Generation from Steering Inequality Violations Using Qudits}},}\
  }\href {\doibase 10.1103/PhysRevLett.120.260401} {\bibfield  {journal}
  {\bibinfo  {journal} {Phys. Rev. Lett.}\ }\textbf {\bibinfo {volume}
  {120}},\ \bibinfo {pages} {260401} (\bibinfo {year} {2018})},\ \Eprint
  {http://arxiv.org/abs/1803.05199} {arXiv:1803.05199} \BibitemShut {NoStop}%
\bibitem [{\citenamefont {Kogias}\ \emph {et~al.}(2017)\citenamefont {Kogias},
  \citenamefont {Xiang}, \citenamefont {He},\ and\ \citenamefont
  {Adesso}}]{Kogias2017}%
  \BibitemOpen
  \bibfield  {author} {\bibinfo {author} {\bibfnamefont {Ioannis}\ \bibnamefont
  {Kogias}}, \bibinfo {author} {\bibfnamefont {Yu}~\bibnamefont {Xiang}},
  \bibinfo {author} {\bibfnamefont {Qiongyi}\ \bibnamefont {He}}, and\
  \bibinfo {author} {\bibfnamefont {Gerardo}\ \bibnamefont {Adesso}},\
  }\bibfield  {title} {\enquote {\bibinfo {title} {{Unconditional security of
  entanglement-based continuous-variable quantum secret sharing}},}\ }\href
  {\doibase 10.1103/PhysRevA.95.012315} {\bibfield  {journal} {\bibinfo
  {journal} {Phys. Rev. A}\ }\textbf {\bibinfo {volume} {95}},\ \bibinfo
  {pages} {012315} (\bibinfo {year} {2017})},\ \Eprint
  {http://arxiv.org/abs/1603.03224} {arXiv:1603.03224} \BibitemShut {NoStop}%
\bibitem [{\citenamefont {Xiang}\ \emph {et~al.}(2017)\citenamefont {Xiang},
  \citenamefont {Kogias}, \citenamefont {Adesso},\ and\ \citenamefont
  {He}}]{Xiang2017}%
  \BibitemOpen
  \bibfield  {author} {\bibinfo {author} {\bibfnamefont {Yu}~\bibnamefont
  {Xiang}}, \bibinfo {author} {\bibfnamefont {Ioannis}\ \bibnamefont {Kogias}},
  \bibinfo {author} {\bibfnamefont {Gerardo}\ \bibnamefont {Adesso}}, and\
  \bibinfo {author} {\bibfnamefont {Qiongyi}\ \bibnamefont {He}},\ }\bibfield
  {title} {\enquote {\bibinfo {title} {{Multipartite Gaussian steering:
  Monogamy constraints and quantum cryptography applications}},}\ }\href
  {\doibase 10.1103/PhysRevA.95.010101} {\bibfield  {journal} {\bibinfo
  {journal} {Phys. Rev. A}\ }\textbf {\bibinfo {volume} {95}},\ \bibinfo
  {pages} {010101(R)} (\bibinfo {year} {2017})},\ \Eprint
  {http://arxiv.org/abs/1603.08173} {arXiv:1603.08173} \BibitemShut {NoStop}%
\bibitem [{\citenamefont {Huang}\ \emph {et~al.}(2019)\citenamefont {Huang},
  \citenamefont {Lambert}, \citenamefont {Li}, \citenamefont {Lu},\ and\
  \citenamefont {Nori}}]{Huang2019}%
  \BibitemOpen
  \bibfield  {author} {\bibinfo {author} {\bibfnamefont {Chien-Ying}\
  \bibnamefont {Huang}}, \bibinfo {author} {\bibfnamefont {Neill}\ \bibnamefont
  {Lambert}}, \bibinfo {author} {\bibfnamefont {Che-Ming}\ \bibnamefont {Li}},
  \bibinfo {author} {\bibfnamefont {Yen-Te}\ \bibnamefont {Lu}}, and\
  \bibinfo {author} {\bibfnamefont {Franco}\ \bibnamefont {Nori}},\ }\bibfield
  {title} {\enquote {\bibinfo {title} {{Securing quantum networking tasks with
  multipartite Einstein-Podolsky-Rosen steering}},}\ }\href {\doibase 10.1103/PhysRevA.99.012302} {\bibfield  {journal} {\bibinfo  {journal}
  {Phys. Rev. A}\ }\textbf {\bibinfo {volume} {99}},\ \bibinfo {pages}
  {012302} (\bibinfo {year} {2019})},\ \Eprint
  {http://arxiv.org/abs/1812.03251} {arXiv:1812.03251} \BibitemShut {NoStop}%
\bibitem [{\citenamefont {Gallego}\ \emph {et~al.}(2012)\citenamefont
  {Gallego}, \citenamefont {W{\"{u}}rflinger}, \citenamefont {Ac{\'{i}}n},\
  and\ \citenamefont {Navascu{\'{e}}s}}]{Gallego2012}%
  \BibitemOpen
  \bibfield  {author} {\bibinfo {author} {\bibfnamefont {Rodrigo}\ \bibnamefont
  {Gallego}}, \bibinfo {author} {\bibfnamefont {Lars~Erik}\ \bibnamefont
  {W{\"{u}}rflinger}}, \bibinfo {author} {\bibfnamefont {Antonio}\ \bibnamefont
  {Ac{\'{i}}n}}, and\ \bibinfo {author} {\bibfnamefont {Miguel}\ \bibnamefont
  {Navascu{\'{e}}s}},\ }\bibfield  {title} {\enquote {\bibinfo {title}
  {{Operational Framework for Nonlocality}},}\ }\href {\doibase 10.1103/PhysRevLett.109.070401} {\bibfield  {journal} {\bibinfo  {journal}
  {Phys. Rev. Lett.}\ }\textbf {\bibinfo {volume} {109}},\ \bibinfo
  {pages} {070401} (\bibinfo {year} {2012})},\ \Eprint
  {http://arxiv.org/abs/1112.2647} {arXiv:1112.2647} \BibitemShut {NoStop}%
\bibitem [{\citenamefont {Bancal}\ \emph {et~al.}(2013)\citenamefont {Bancal},
  \citenamefont {Barrett}, \citenamefont {Gisin},\ and\ \citenamefont
  {Pironio}}]{Bancal2013}%
  \BibitemOpen
  \bibfield  {author} {\bibinfo {author} {\bibfnamefont {Jean-Daniel}\
  \bibnamefont {Bancal}}, \bibinfo {author} {\bibfnamefont {Jonathan}\
  \bibnamefont {Barrett}}, \bibinfo {author} {\bibfnamefont {Nicolas}\
  \bibnamefont {Gisin}}, and\ \bibinfo {author} {\bibfnamefont {Stefano}\
  \bibnamefont {Pironio}},\ }\bibfield  {title} {\enquote {\bibinfo {title}
  {{Definitions of multipartite nonlocality}},}\ }\href {\doibase 10.1103/PhysRevA.88.014102} {\bibfield  {journal} {\bibinfo  {journal}
  {Phys. Rev. A}\ }\textbf {\bibinfo {volume} {88}},\ \bibinfo {pages}
  {14102} (\bibinfo {year} {2013})},\ \Eprint {http://arxiv.org/abs/1112.2626}
  {arXiv:1112.2626} \BibitemShut {NoStop}%
\bibitem [{\citenamefont {Brand{\~{a}}o}\ and\ \citenamefont
  {Gour}(2015{\natexlab{a}})}]{Brandao2015}%
  \BibitemOpen
  \bibfield  {author} {\bibinfo {author} {\bibfnamefont {Fernando G. S.~L.}\
  \bibnamefont {Brand{\~{a}}o}}\ and\ \bibinfo {author} {\bibfnamefont {Gilad}\
  \bibnamefont {Gour}},\ }\bibfield  {title} {\enquote {\bibinfo {title}
  {{Reversible Framework for Quantum Resource Theories}},}\ }\href {\doibase 10.1103/PhysRevLett.115.070503} {\bibfield  {journal} {\bibinfo  {journal}
  {Phys. Rev. Lett.}\ }\textbf {\bibinfo {volume} {115}},\ \bibinfo
  {pages} {70503} (\bibinfo {year} {2015}{\natexlab{a}})}, \ \Eprint {http://arxiv.org/abs/1502.03149} {arXiv:1502.03149}
	\BibitemShut {NoStop}%
\bibitem [{\citenamefont {Brand{\~{a}}o}\ and\ \citenamefont
  {Gour}(2015{\natexlab{b}})}]{Brandao2015a}%
  \BibitemOpen
  \bibfield  {author} {\bibinfo {author} {\bibfnamefont {Fernando G. S.~L.}\
  \bibnamefont {Brand{\~{a}}o}}\ and\ \bibinfo {author} {\bibfnamefont {Gilad}\
  \bibnamefont {Gour}},\ }\bibfield  {title} {\enquote {\bibinfo {title}
  {{Erratum: Reversible Framework for Quantum Resource Theories [Phys. Rev.
  Lett. 115 , 070503 (2015)]}},}\ }\href {\doibase 10.1103/PhysRevLett.115.199901} {\bibfield  {journal} {\bibinfo  {journal}
  {Phys. Rev. Lett.}\ }\textbf {\bibinfo {volume} {115}},\ \bibinfo
  {pages} {199901(E)} (\bibinfo {year} {2015}{\natexlab{b}})}
	\BibitemShut {NoStop}%
\bibitem [{\citenamefont {Coecke}\ \emph {et~al.}(2016)\citenamefont {Coecke},
  \citenamefont {Fritz},\ and\ \citenamefont {Spekkens}}]{Coecke2016}%
  \BibitemOpen
  \bibfield  {author} {\bibinfo {author} {\bibfnamefont {Bob}\ \bibnamefont
  {Coecke}}, \bibinfo {author} {\bibfnamefont {Tobias}\ \bibnamefont {Fritz}}, and\ \bibinfo {author} {\bibfnamefont {Robert~W.}\ \bibnamefont
  {Spekkens}},\ }\bibfield  {title} {\enquote {\bibinfo {title} {{A
  mathematical theory of resources}},}\ }\href {\doibase 10.1016/j.ic.2016.02.008} {\bibfield  {journal} {\bibinfo  {journal} {Inf. Comput.}\ }\textbf {\bibinfo {volume} {250}},\ \bibinfo
  {pages} {59--86} (\bibinfo {year} {2016})},\ \Eprint
  {http://arxiv.org/abs/1409.5531} {arXiv:1409.5531} \BibitemShut {NoStop}%
\bibitem [{\citenamefont {de~Vicente}(2014)}]{deVicente2014}%
  \BibitemOpen
  \bibfield  {author} {\bibinfo {author} {\bibfnamefont {Julio~I}\ \bibnamefont
  {de~Vicente}},\ }\bibfield  {title} {\enquote {\bibinfo {title} {{On
  nonlocality as a resource theory and nonlocality measures}},}\ }\href
  {\doibase 10.1088/1751-8113/47/42/424017} {\bibfield  {journal} {\bibinfo
  {journal} {J. Phys. A Math. Theor.}\ }\textbf
  {\bibinfo {volume} {47}},\ \bibinfo {pages} {424017} (\bibinfo {year}
  {2014})},\ \Eprint {http://arxiv.org/abs/1401.6941} {arXiv:1401.6941}
  \BibitemShut {NoStop}%
\bibitem [{\citenamefont {Gallego}\ and\ \citenamefont
  {Aolita}(2017)}]{Gallego2017}%
  \BibitemOpen
  \bibfield  {author} {\bibinfo {author} {\bibfnamefont {Rodrigo}\ \bibnamefont
  {Gallego}}\ and\ \bibinfo {author} {\bibfnamefont {Leandro}\ \bibnamefont
  {Aolita}},\ }\bibfield  {title} {\enquote {\bibinfo {title} {{Nonlocality
  free wirings and the distinguishability between Bell boxes}},}\ }\href
  {\doibase 10.1103/PhysRevA.95.032118} {\bibfield  {journal} {\bibinfo
  {journal} {Phys. Rev. A}\ }\textbf {\bibinfo {volume} {95}},\ \bibinfo
  {pages} {032118} (\bibinfo {year} {2017})},\ \Eprint
  {http://arxiv.org/abs/1611.06932} {arXiv:1611.06932} \BibitemShut {NoStop}%
\bibitem [{\citenamefont {Wolfe}\ \emph {et~al.}(2020)\citenamefont {Wolfe},
  \citenamefont {Schmid}, \citenamefont {Sainz}, \citenamefont {Kunjwal},\ and\
  \citenamefont {Spekkens}}]{Wolfe2019}%
  \BibitemOpen
  \bibfield  {author} {\bibinfo {author} {\bibfnamefont {Elie}\ \bibnamefont
  {Wolfe}}, \bibinfo {author} {\bibfnamefont {David}\ \bibnamefont {Schmid}},
  \bibinfo {author} {\bibfnamefont {Ana~Bel{\'{e}}n}\ \bibnamefont {Sainz}},
  \bibinfo {author} {\bibfnamefont {Ravi}\ \bibnamefont {Kunjwal}}, and\
  \bibinfo {author} {\bibfnamefont {Robert~W.}\ \bibnamefont {Spekkens}},\
  }\bibfield  {title} {\enquote {\bibinfo {title} {{Quantifying Bell: the
  Resource Theory of Nonclassicality of Common-Cause Boxes}},}\ }\href
  {\doibase 10.22331/q-2020-06-08-280} {\bibfield  {journal} {\bibinfo
  {journal} {Quantum}\ }\textbf {\bibinfo {volume} {4}},\ \bibinfo {pages}
  {280} (\bibinfo {year} {2020})},\ \Eprint {http://arxiv.org/abs/1903.06311}
  {arXiv:1903.06311} \BibitemShut {NoStop}%
\bibitem [{\citenamefont {Winter}\ and\ \citenamefont
  {Yang}(2016)}]{Winter2016}%
  \BibitemOpen
  \bibfield  {author} {\bibinfo {author} {\bibfnamefont {Andreas}\ \bibnamefont
  {Winter}}\ and\ \bibinfo {author} {\bibfnamefont {Dong}\ \bibnamefont
  {Yang}},\ }\bibfield  {title} {\enquote {\bibinfo {title} {{Operational
  Resource Theory of Coherence}},}\ }\href {\doibase 10.1103/PhysRevLett.116.120404} {\bibfield  {journal} {\bibinfo  {journal}
  {Phys. Rev. Lett.}\ }\textbf {\bibinfo {volume} {116}},\ \bibinfo
  {pages} {120404} (\bibinfo {year} {2016})},\ \Eprint
  {http://arxiv.org/abs/1506.07975} {arXiv:1506.07975} \BibitemShut {NoStop}%
\bibitem [{\citenamefont {Chitambar}\ and\ \citenamefont
  {Gour}(2016)}]{Chitambar2016}%
  \BibitemOpen
  \bibfield  {author} {\bibinfo {author} {\bibfnamefont {Eric}\ \bibnamefont
  {Chitambar}}\ and\ \bibinfo {author} {\bibfnamefont {Gilad}\ \bibnamefont
  {Gour}},\ }\bibfield  {title} {\enquote {\bibinfo {title} {{Critical
  Examination of Incoherent Operations and a Physically Consistent Resource
  Theory of Quantum Coherence}},}\ }\href {\doibase 10.1103/PhysRevLett.117.030401} {\bibfield  {journal} {\bibinfo  {journal}
  {Phys. Rev. Lett.}\ }\textbf {\bibinfo {volume} {117}},\ \bibinfo
  {pages} {030401} (\bibinfo {year} {2016})},\ \Eprint
  {http://arxiv.org/abs/1602.06969} {arXiv:1602.06969} \BibitemShut {NoStop}%
\bibitem [{\citenamefont {Grudka}\ \emph {et~al.}(2014)\citenamefont {Grudka},
  \citenamefont {Horodecki}, \citenamefont {Horodecki}, \citenamefont
  {Horodecki}, \citenamefont {Horodecki}, \citenamefont {Joshi}, \citenamefont
  {K{\l}obus},\ and\ \citenamefont {W{\'{o}}jcik}}]{Grudka2013}%
  \BibitemOpen
  \bibfield  {author} {\bibinfo {author} {\bibfnamefont {A.}~\bibnamefont
  {Grudka}} et al, }\bibfield  {title} {\enquote {\bibinfo
  {title} {{Quantifying Contextuality}},}\ }\href {\doibase 10.1103/PhysRevLett.112.120401} {\bibfield  {journal} {\bibinfo  {journal}
  {Phys. Rev. Lett.}\ }\textbf {\bibinfo {volume} {112}},\ \bibinfo
  {pages} {120401} (\bibinfo {year} {2014})},\ \Eprint
  {http://arxiv.org/abs/1209.3745} {arXiv:1209.3745} \BibitemShut {NoStop}%
\bibitem [{\citenamefont {Amaral}\ \emph {et~al.}(2018)\citenamefont {Amaral},
  \citenamefont {Cabello}, \citenamefont {Cunha},\ and\ \citenamefont
  {Aolita}}]{Amaral2018}%
  \BibitemOpen
  \bibfield  {author} {\bibinfo {author} {\bibfnamefont {Barbara}\ \bibnamefont
  {Amaral}}, \bibinfo {author} {\bibfnamefont {Ad{\'{a}}n}\ \bibnamefont
  {Cabello}}, \bibinfo {author} {\bibfnamefont {Marcelo~Terra}\ \bibnamefont
  {Cunha}}, and\ \bibinfo {author} {\bibfnamefont {Leandro}\ \bibnamefont
  {Aolita}},\ }\bibfield  {title} {\enquote {\bibinfo {title} {{Noncontextual
  Wirings}},}\ }\href {\doibase 10.1103/PhysRevLett.120.130403} {\bibfield
  {journal} {\bibinfo  {journal} {Phys. Rev. Lett.}\ }\textbf {\bibinfo
  {volume} {120}},\ \bibinfo {pages} {130403} (\bibinfo {year}
  {2018})},\ \Eprint{http://arxiv.org/abs/1705.07911} {arXiv:1705.07911}
	\BibitemShut {NoStop}%
\bibitem [{\citenamefont {Taddei}\ \emph {et~al.}(2019)\citenamefont {Taddei},
  \citenamefont {Nery},\ and\ \citenamefont {Aolita}}]{Taddei2019}%
  \BibitemOpen
  \bibfield  {author} {\bibinfo {author} {\bibfnamefont {M{\'{a}}rcio~M.}\
  \bibnamefont {Taddei}}, \bibinfo {author} {\bibfnamefont {Ranieri~V.}\
  \bibnamefont {Nery}}, and\ \bibinfo {author} {\bibfnamefont {Leandro}\
  \bibnamefont {Aolita}},\ }\bibfield  {title} {\enquote {\bibinfo {title}
  {{Quantum superpositions of causal orders as an operational resource}},}\
  }\href {\doibase 10.1103/PhysRevResearch.1.033174} {\bibfield  {journal}
  {\bibinfo  {journal} {Phys. Rev. Res.}\ }\textbf {\bibinfo {volume}
  {1}},\ \bibinfo {pages} {033174} (\bibinfo {year} {2019})},\ \Eprint
  {http://arxiv.org/abs/1903.06180} {arXiv:1903.06180} \BibitemShut {NoStop}%
\bibitem [{\citenamefont {Gallego}\ and\ \citenamefont
  {Aolita}(2015)}]{Gallego2015}%
  \BibitemOpen
  \bibfield  {author} {\bibinfo {author} {\bibfnamefont {Rodrigo}\ \bibnamefont
  {Gallego}}\ and\ \bibinfo {author} {\bibfnamefont {Leandro}\ \bibnamefont
  {Aolita}},\ }\bibfield  {title} {\enquote {\bibinfo {title} {{Resource Theory
  of Steering}},}\ }\href {\doibase 10.1103/PhysRevX.5.041008} {\bibfield
  {journal} {\bibinfo  {journal} {Phys. Rev. X}\ }\textbf {\bibinfo
  {volume} {5}},\ \bibinfo {pages} {041008} (\bibinfo {year} {2015})},\ \Eprint
  {http://arxiv.org/abs/1409.5804} {arXiv:1409.5804} \BibitemShut {NoStop}%
\bibitem [{\citenamefont {Kaur}\ and\ \citenamefont {Wilde}(2017)}]{Kaur2017}%
  \BibitemOpen
  \bibfield  {author} {\bibinfo {author} {\bibfnamefont {Eneet}\ \bibnamefont
  {Kaur}}\ and\ \bibinfo {author} {\bibfnamefont {Mark~M.}\ \bibnamefont
  {Wilde}},\ }\bibfield  {title} {\enquote {\bibinfo {title} {{Relative entropy
  of steering: on its definition and properties}},}\ }\href {\doibase 10.1088/1751-8121/aa907b} {\bibfield  {journal} {\bibinfo  {journal} {J. Phys. A Math. Theor.}\ }\textbf {\bibinfo {volume}
  {50}},\ \bibinfo {pages} {465301} (\bibinfo {year} {2017})},\ \Eprint
  {http://arxiv.org/abs/1612.07152} {arXiv:1612.07152} \BibitemShut {NoStop}%
\bibitem [{\citenamefont {Bell}(1964)}]{Bell1964}%
  \BibitemOpen
  \bibfield  {author} {\bibinfo {author} {\bibfnamefont {John~S.}\ \bibnamefont
  {Bell}},\ }\bibfield  {title} {\enquote {\bibinfo {title} {{On the Einstein
  Podolsky Rosen paradox}},}\ }\href {\doibase 10.1103/PhysicsPhysiqueFizika.1.195} {\bibfield  {journal} {\bibinfo
  {journal} {Phys. Phys. Fiz.}\ }\textbf {\bibinfo {volume} {1}},\
  \bibinfo {pages} {195--200} (\bibinfo {year} {1964})}\BibitemShut {NoStop}%
\bibitem [{\citenamefont {Svetlichny}(1987)}]{Svetlichny1987}%
  \BibitemOpen
  \bibfield  {author} {\bibinfo {author} {\bibfnamefont {George}\ \bibnamefont
  {Svetlichny}},\ }\bibfield  {title} {\enquote {\bibinfo {title}
  {{Distinguishing three-body from two-body nonseparability by a Bell-type
  inequality}},}\ }\href {\doibase 10.1103/PhysRevD.35.3066} {\bibfield
  {journal} {\bibinfo  {journal} {Phys. Rev. D}\ }\textbf {\bibinfo
  {volume} {35}},\ \bibinfo {pages} {3066--3069} (\bibinfo {year}
  {1987})}\BibitemShut {NoStop}%
\bibitem [{\citenamefont {Wood}\ and\ \citenamefont
  {Spekkens}(2015)}]{Wood2015}%
  \BibitemOpen
  \bibfield  {author} {\bibinfo {author} {\bibfnamefont {Christopher~J.}\
  \bibnamefont {Wood}}\ and\ \bibinfo {author} {\bibfnamefont {Robert~W.}\
  \bibnamefont {Spekkens}},\ }\bibfield  {title} {\enquote {\bibinfo {title}
  {{The lesson of causal discovery algorithms for quantum correlations: causal
  explanations of Bell-inequality violations require fine-tuning}},}\ }\href
  {\doibase 10.1088/1367-2630/17/3/033002} {\bibfield  {journal} {\bibinfo
  {journal} {New J. Phys.}\ }\textbf {\bibinfo {volume} {17}},\
  \bibinfo {pages} {033002} (\bibinfo {year} {2015})},\ \Eprint
  {http://arxiv.org/abs/1208.4119} {arXiv:1208.4119} \BibitemShut {NoStop}%
\bibitem [{\citenamefont {Cavalcanti}\ \emph {et~al.}(2011)\citenamefont
  {Cavalcanti}, \citenamefont {He}, \citenamefont {Reid},\ and\ \citenamefont
  {Wiseman}}]{Cavalcanti2011}%
  \BibitemOpen
  \bibfield  {author} {\bibinfo {author} {\bibfnamefont {E.~G.}\ \bibnamefont
  {Cavalcanti}}, \bibinfo {author} {\bibfnamefont {Q.~Y.}\ \bibnamefont {He}},
  \bibinfo {author} {\bibfnamefont {M.~D.}\ \bibnamefont {Reid}}, and\
  \bibinfo {author} {\bibfnamefont {H.~M.}\ \bibnamefont {Wiseman}},\
  }\bibfield  {title} {\enquote {\bibinfo {title} {{Unified criteria for
  multipartite quantum nonlocality}},}\ }\href {\doibase 10.1103/PhysRevA.84.032115} {\bibfield  {journal} {\bibinfo  {journal}
  {Phys. Rev. A}\ }\textbf {\bibinfo {volume} {84}},\ \bibinfo {pages}
  {032115} (\bibinfo {year} {2011})},\ \Eprint {http://arxiv.org/abs/1008.5014}
  {arXiv:1008.5014} \BibitemShut {NoStop}%
\bibitem [{\citenamefont {Armstrong}\ \emph {et~al.}(2015)\citenamefont
  {Armstrong}, \citenamefont {Wang}, \citenamefont {Teh}, \citenamefont {Gong},
  \citenamefont {He}, \citenamefont {Janousek}, \citenamefont {Bachor},
  \citenamefont {Reid},\ and\ \citenamefont {Lam}}]{Armstrong2015}%
  \BibitemOpen
  \bibfield  {author} {\bibinfo {author} {\bibfnamefont {Seiji}\ \bibnamefont
  {Armstrong}} et al, }\bibfield  {title}
  {\enquote {\bibinfo {title} {{Multipartite Einstein–Podolsky–Rosen
  steering and genuine tripartite entanglement with optical networks}},}\
  }\href {\doibase 10.1038/nphys3202} {\bibfield  {journal} {\bibinfo
  {journal} {Nat. Phys.}\ }\textbf {\bibinfo {volume} {11}},\ \bibinfo
  {pages} {167--172} (\bibinfo {year} {2015})},\ \Eprint
  {http://arxiv.org/abs/1412.7212} {arXiv:1412.7212} \BibitemShut {NoStop}%
\bibitem [{\citenamefont {Taddei}\ \emph {et~al.}(2016)\citenamefont {Taddei},
  \citenamefont {Nery},\ and\ \citenamefont {Aolita}}]{Taddei2016}%
  \BibitemOpen
  \bibfield  {author} {\bibinfo {author} {\bibfnamefont {M.~M.}\ \bibnamefont
  {Taddei}}, \bibinfo {author} {\bibfnamefont {R.~V.}\ \bibnamefont {Nery}},
  and\ \bibinfo {author} {\bibfnamefont {L.}~\bibnamefont {Aolita}},\
  }\bibfield  {title} {\enquote {\bibinfo {title} {{Necessary and sufficient
  conditions for multipartite Bell violations with only one trusted device}},}\
  }\href {\doibase 10.1103/PhysRevA.94.032106} {\bibfield  {journal} {\bibinfo
  {journal} {Phys. Rev. A}\ }\textbf {\bibinfo {volume} {94}},\ \bibinfo
  {pages} {032106} (\bibinfo {year} {2016})},\ \Eprint
  {http://arxiv.org/abs/1603.05247} {arXiv:1603.05247} \BibitemShut {NoStop}%
\bibitem [{\citenamefont {Li}\ \emph {et~al.}(2015)\citenamefont {Li},
  \citenamefont {Chen}, \citenamefont {Chen}, \citenamefont {Zhang},
  \citenamefont {Chen},\ and\ \citenamefont {Pan}}]{Li2015}%
  \BibitemOpen
  \bibfield  {author} {\bibinfo {author} {\bibfnamefont {Che-Ming}\
  \bibnamefont {Li}} et al, }\bibfield  {title}
  {\enquote {\bibinfo {title} {{Genuine High-Order Einstein-Podolsky-Rosen
  Steering}},}\ }\href {\doibase 10.1103/PhysRevLett.115.010402} {\bibfield
  {journal} {\bibinfo  {journal} {Phys. Rev. Lett.}\ }\textbf {\bibinfo
  {volume} {115}},\ \bibinfo {pages} {010402} (\bibinfo {year} {2015})},\
  \Eprint {http://arxiv.org/abs/1501.01452} {arXiv:1501.01452} \BibitemShut
  {NoStop}%
\bibitem [{\citenamefont {Cavalcanti}\ \emph {et~al.}(2015)\citenamefont
  {Cavalcanti}, \citenamefont {Skrzypczyk}, \citenamefont {Aguilar},
  \citenamefont {Nery}, \citenamefont {Ribeiro},\ and\ \citenamefont
  {Walborn}}]{Cavalcanti2015a}%
  \BibitemOpen
  \bibfield  {author} {\bibinfo {author} {\bibfnamefont {D.}~\bibnamefont
  {Cavalcanti}} et al, }\bibfield  {title} {\enquote {\bibinfo {title} {{Detection of entanglement
  in asymmetric quantum networks and multipartite quantum steering}},}\ }\href
  {\doibase 10.1038/ncomms8941} {\bibfield  {journal} {\bibinfo  {journal}
  {Nat. Commun.}\ }\textbf {\bibinfo {volume} {6}},\ \bibinfo {pages}
  {7941} (\bibinfo {year} {2015})},\ \Eprint {http://arxiv.org/abs/1412.7730}
  {arXiv:1412.7730} \BibitemShut {NoStop}%
\bibitem [{\citenamefont {Sainz}\ \emph {et~al.}(2015)\citenamefont {Sainz},
  \citenamefont {Brunner}, \citenamefont {Cavalcanti}, \citenamefont
  {Skrzypczyk},\ and\ \citenamefont {V{\'{e}}rtesi}}]{Sainz2015}%
  \BibitemOpen
  \bibfield  {author} {\bibinfo {author} {\bibfnamefont {Ana~Bel{\'{e}}n}\
  \bibnamefont {Sainz}}, \bibinfo {author} {\bibfnamefont {Nicolas}\
  \bibnamefont {Brunner}}, \bibinfo {author} {\bibfnamefont {Daniel}\
  \bibnamefont {Cavalcanti}}, \bibinfo {author} {\bibfnamefont {Paul}\
  \bibnamefont {Skrzypczyk}}, and\ \bibinfo {author} {\bibfnamefont
  {Tam{\'{a}}s}\ \bibnamefont {V{\'{e}}rtesi}},\ }\bibfield  {title} {\enquote
  {\bibinfo {title} {{Postquantum Steering}},}\ }\href {\doibase 10.1103/PhysRevLett.115.190403} {\bibfield  {journal} {\bibinfo  {journal}
  {Phys. Rev. Lett.}\ }\textbf {\bibinfo {volume} {115}},\ \bibinfo
  {pages} {190403} (\bibinfo {year} {2015})},\ \Eprint
  {http://arxiv.org/abs/1505.01430} {arXiv:1505.01430} \BibitemShut {NoStop}%
\bibitem [{\citenamefont {Sainz}\ \emph {et~al.}(2018)\citenamefont {Sainz},
  \citenamefont {Aolita}, \citenamefont {Piani}, \citenamefont {Hoban},\ and\
  \citenamefont {Skrzypczyk}}]{Sainz2018a}%
  \BibitemOpen
  \bibfield  {author} {\bibinfo {author} {\bibfnamefont {A.~B.}\ \bibnamefont
  {Sainz}}, \bibinfo {author} {\bibfnamefont {L.}~\bibnamefont {Aolita}},
  \bibinfo {author} {\bibfnamefont {M.}~\bibnamefont {Piani}}, \bibinfo
  {author} {\bibfnamefont {M.~J.}\ \bibnamefont {Hoban}}, and\ \bibinfo
  {author} {\bibfnamefont {P.}~\bibnamefont {Skrzypczyk}},\ }\bibfield  {title}
  {\enquote {\bibinfo {title} {{A formalism for steering with local quantum
  measurements}},}\ }\href {\doibase 10.1088/1367-2630/aad8df} {\bibfield
  {journal} {\bibinfo  {journal} {New J. Phys.}\ }\textbf {\bibinfo
  {volume} {20}},\ \bibinfo {pages} {083040} (\bibinfo {year} {2018})},\
  \Eprint {http://arxiv.org/abs/1708.00756} {arXiv:1708.00756} \BibitemShut
  {NoStop}%
\bibitem [{\citenamefont {Sainz}\ \emph {et~al.}(2020)\citenamefont {Sainz},
  \citenamefont {Hoban}, \citenamefont {Skrzypczyk},\ and\ \citenamefont
  {Aolita}}]{Sainz2019}%
  \BibitemOpen
  \bibfield  {author} {\bibinfo {author} {\bibfnamefont {Ana~Bel{\'{e}}n}\
  \bibnamefont {Sainz}}, \bibinfo {author} {\bibfnamefont {Matty~J.}\
  \bibnamefont {Hoban}}, \bibinfo {author} {\bibfnamefont {Paul}\ \bibnamefont
  {Skrzypczyk}}, and\ \bibinfo {author} {\bibfnamefont {Leandro}\
  \bibnamefont {Aolita}},\ }\bibfield  {title} {\enquote {\bibinfo {title}
  {{Bipartite Postquantum Steering in Generalized Scenarios}},}\ }\href
  {\doibase 10.1103/PhysRevLett.125.050404} {\bibfield  {journal} {\bibinfo
  {journal} {Phys. Rev. Lett.}\ }\textbf {\bibinfo {volume} {125}},\
  \bibinfo {pages} {050404} (\bibinfo {year} {2020})},\ \Eprint
  {http://arxiv.org/abs/1907.03705} {arXiv:1907.03705} \BibitemShut {NoStop}%
\bibitem [{\citenamefont {Popescu}\ and\ \citenamefont
  {Rohrlich}(1994)}]{Popescu1994}%
  \BibitemOpen
  \bibfield  {author} {\bibinfo {author} {\bibfnamefont {Sandu}\ \bibnamefont
  {Popescu}}\ and\ \bibinfo {author} {\bibfnamefont {Daniel}\ \bibnamefont
  {Rohrlich}},\ }\bibfield  {title} {\enquote {\bibinfo {title} {{Quantum
  nonlocality as an axiom}},}\ }\href {\doibase 10.1007/BF02058098} {\bibfield
  {journal} {\bibinfo  {journal} {Found. Phys.}\ }\textbf {\bibinfo
  {volume} {24}},\ \bibinfo {pages} {379--385} (\bibinfo {year}
  {1994})},\ \Eprint{http://arxiv.org/abs/quant-ph/9508009} {arXiv:quant-ph/9508009}
	\BibitemShut {NoStop}%
\bibitem [{\citenamefont {Kwiat}\ \emph {et~al.}(1999)\citenamefont {Kwiat},
  \citenamefont {Waks}, \citenamefont {White}, \citenamefont {Appelbaum},\ and\
  \citenamefont {Eberhard}}]{Kwiat1999}%
  \BibitemOpen
  \bibfield  {author} {\bibinfo {author} {\bibfnamefont {Paul~G.}\ \bibnamefont
  {Kwiat}}, \bibinfo {author} {\bibfnamefont {Edo}\ \bibnamefont {Waks}},
  \bibinfo {author} {\bibfnamefont {Andrew~G.}\ \bibnamefont {White}}, \bibinfo
  {author} {\bibfnamefont {Ian}\ \bibnamefont {Appelbaum}}, and\ \bibinfo
  {author} {\bibfnamefont {Philippe~H.}\ \bibnamefont {Eberhard}},\ }\bibfield
  {title} {\enquote {\bibinfo {title} {{Ultrabright source of
  polarization-entangled photons}},}\ }\href {\doibase 10.1103/PhysRevA.60.R773} {\bibfield  {journal} {\bibinfo  {journal}
  {Phys. Rev. A}\ }\textbf {\bibinfo {volume} {60}},\ \bibinfo {pages}
  {R773--R776} (\bibinfo {year} {1999})},\ \Eprint
  {http://arxiv.org/abs/quant-ph/9810003} {arXiv:quant-ph/9810003 } \BibitemShut
  {NoStop}%
\bibitem [{\citenamefont {Far{\'{i}}as}\ \emph {et~al.}(2012)\citenamefont
  {Far{\'{i}}as}, \citenamefont {Aguilar}, \citenamefont
  {Vald{\'{e}}s-Hern{\'{a}}ndez}, \citenamefont {Ribeiro}, \citenamefont
  {Davidovich},\ and\ \citenamefont {Walborn}}]{Farias2012}%
  \BibitemOpen
  \bibfield  {author} {\bibinfo {author} {\bibfnamefont {O.~Jim{\'{e}}nez}\
  \bibnamefont {Far{\'{i}}as}} et al, }\bibfield  {title} {\enquote {\bibinfo
  {title} {{Observation of the Emergence of Multipartite Entanglement Between a
  Bipartite System and its Environment}},}\ }\href {\doibase 10.1103/PhysRevLett.109.150403} {\bibfield  {journal} {\bibinfo  {journal}
  {Phys. Rev. Lett.}\ }\textbf {\bibinfo {volume} {109}},\ \bibinfo
  {pages} {150403} (\bibinfo {year} {2012})}\BibitemShut {NoStop}%
\bibitem [{\citenamefont {Sainz}\ \emph {et~al.}(2016)\citenamefont {Sainz},
  \citenamefont {Aolita}, \citenamefont {Brunner}, \citenamefont {Gallego},\
  and\ \citenamefont {Skrzypczyk}}]{Sainz2016a}%
  \BibitemOpen
  \bibfield  {author} {\bibinfo {author} {\bibfnamefont {Ana~Bel{\'{e}}n}\
  \bibnamefont {Sainz}}, \bibinfo {author} {\bibfnamefont {Leandro}\
  \bibnamefont {Aolita}}, \bibinfo {author} {\bibfnamefont {Nicolas}\
  \bibnamefont {Brunner}}, \bibinfo {author} {\bibfnamefont {Rodrigo}\
  \bibnamefont {Gallego}}, and\ \bibinfo {author} {\bibfnamefont {Paul}\
  \bibnamefont {Skrzypczyk}},\ }\bibfield  {title} {\enquote {\bibinfo {title}
  {{Classical communication cost of quantum steering}},}\ }\href {\doibase 10.1103/PhysRevA.94.012308} {\bibfield  {journal} {\bibinfo  {journal}
  {Phys. Rev. A}\ }\textbf {\bibinfo {volume} {94}},\ \bibinfo {pages}
  {012308} (\bibinfo {year} {2016})},\ \Eprint
  {http://arxiv.org/abs/1603.05079} {arXiv:1603.05079} \BibitemShut {NoStop}%
\bibitem [{\citenamefont {Clauser}\ \emph {et~al.}(1969)\citenamefont
  {Clauser}, \citenamefont {Horne}, \citenamefont {Shimony},\ and\
  \citenamefont {Holt}}]{Clauser1969}%
  \BibitemOpen
  \bibfield  {author} {\bibinfo {author} {\bibfnamefont {John~F.}\ \bibnamefont
  {Clauser}}, \bibinfo {author} {\bibfnamefont {Michael~A.}\ \bibnamefont
  {Horne}}, \bibinfo {author} {\bibfnamefont {Abner}\ \bibnamefont {Shimony}},
  and\ \bibinfo {author} {\bibfnamefont {Richard~A.}\ \bibnamefont {Holt}},\
  }\bibfield  {title} {\enquote {\bibinfo {title} {{Proposed Experiment to Test
  Local Hidden-Variable Theories}},}\ }\href {\doibase 10.1103/PhysRevLett.23.880} {\bibfield  {journal} {\bibinfo  {journal}
  {Phys. Rev. Lett.}\ }\textbf {\bibinfo {volume} {23}},\ \bibinfo
  {pages} {880--884} (\bibinfo {year} {1969})}\BibitemShut {NoStop}%
\bibitem [{\citenamefont {Gallego}\ \emph {et~al.}(2011)\citenamefont
  {Gallego}, \citenamefont {W{\"{u}}rflinger}, \citenamefont {Ac{\'{i}}n},\
  and\ \citenamefont {Navascu{\'{e}}s}}]{Gallego2011}%
  \BibitemOpen
  \bibfield  {author} {\bibinfo {author} {\bibfnamefont {Rodrigo}\ \bibnamefont
  {Gallego}}, \bibinfo {author} {\bibfnamefont {Lars~Erik}\ \bibnamefont
  {W{\"{u}}rflinger}}, \bibinfo {author} {\bibfnamefont {Antonio}\ \bibnamefont
  {Ac{\'{i}}n}}, and\ \bibinfo {author} {\bibfnamefont {Miguel}\ \bibnamefont
  {Navascu{\'{e}}s}},\ }\bibfield  {title} {\enquote {\bibinfo {title}
  {{Quantum Correlations Require Multipartite Information Principles}},}\
  }\href {\doibase 10.1103/PhysRevLett.107.210403} {\bibfield  {journal}
  {\bibinfo  {journal} {Phys. Rev. Lett.}\ }\textbf {\bibinfo {volume}
  {107}},\ \bibinfo {pages} {210403} (\bibinfo {year} {2011})},\ \Eprint
  {http://arxiv.org/abs/1107.3738} {arXiv:1107.3738} \BibitemShut {NoStop}%
\bibitem [{\citenamefont {Roehsner}\ \emph {et~al.}(2018)\citenamefont {Roehsner}, \citenamefont {Kettlewell}, \citenamefont {Batalh{\~{a}}o}, \citenamefont {Fitzsimons}, and\ \citenamefont {Walther}}]{Roehsner2018}%
  \BibitemOpen
  \bibfield  {author} {\bibinfo {author} {\bibfnamefont {Marie-Christine}\
  \bibnamefont {Roehsner}}, \bibinfo {author} {\bibfnamefont {Joshua~A.}\
  \bibnamefont {Kettlewell}}, \bibinfo {author} {\bibfnamefont {Tiago~B.}\
  \bibnamefont {Batalh{\~{a}}o}}, \bibinfo {author} {\bibfnamefont {Joseph~F.}\
  \bibnamefont {Fitzsimons}}, and\ \bibinfo {author} {\bibfnamefont {Philip}\
  \bibnamefont {Walther}},\ }\bibfield  {title} {\enquote {\bibinfo {title}
  {{Quantum advantage for probabilistic one-time programs}},}\ }\href {\doibase 10.1038/s41467-018-07591-2} {\bibfield  {journal} {\bibinfo  {journal}
  {Nat. Commun.}\ }\textbf {\bibinfo {volume} {9}},\ \bibinfo {pages}
  {5225} (\bibinfo {year} {2018})},\ \Eprint {http://arxiv.org/abs/1709.09724}
  {arXiv:1709.09724} \BibitemShut {NoStop}%
\end{thebibliography}
\end{document}